\documentclass[onecolumn,notitlepage,superscriptaddress]{revtex4-1}

\usepackage{dcolumn}% Align table columns on decimal point
\usepackage{bm}% bold math

\usepackage{amsmath}
\usepackage{graphicx}
\usepackage{amsbsy,amsfonts}
\usepackage{amsthm}
\usepackage{natbib}
\usepackage{colonequals}
\usepackage[colorlinks]{hyperref}
\usepackage{hypcap}
\usepackage{enumitem}

\setlist[enumerate]{itemsep=0mm}

\newcommand{\ket}[1]{\left|#1\right\rangle}

\newcommand{\ketbra}[2]{\left|#1\right\rangle\!\left\langle #2\right|}

\newcommand{\up}[1]{\push{\raisebox{6pt}{$#1$}}}

\newcommand{\suppress}[1]{}

\def\squareforqed{\hbox{\rlap{$\sqcap$}$\sqcup$}}
\def\qed{\ifmmode\squareforqed\else{\unskip\nobreak\hfil
\penalty50\hskip1em\null\nobreak\hfil\squareforqed
\parfillskip=0pt\finalhyphendemerits=0\endgraf}\fi}

\usepackage{algpseudocode}

    \newcommand{\inlinecomment}[1]{\Comment {\footnotesize #1} \normalsize}

\usepackage{multirow}
\usepackage{subcaption}
\DeclareCaptionType{algorithm}
\DeclareCaptionSubType{algorithm}

\newtheorem{theorem}{Theorem}
\newtheorem{lemma}[theorem]{Lemma}

\newtheorem{corollary}[theorem]{Corollary}
\newenvironment{proofof}[1]{\begin{proof}[Proof of~#1.]}{\end{proof}}

\newcommand{\eq}[1]{\hyperref[eq:#1]{(\ref*{eq:#1})}}
\renewcommand{\sec}[1]{\hyperref[sec:#1]{Section~\ref*{sec:#1}}}
\newcommand{\app}[1]{\hyperref[app:#1]{Appendix~\ref*{app:#1}}}
\newcommand{\fig}[1]{\hyperref[fig:#1]{Figure~\ref*{fig:#1}}}
\newcommand{\thm}[1]{\hyperref[thm:#1]{Theorem~\ref*{thm:#1}}}
\newcommand{\lem}[1]{\hyperref[lem:#1]{Lemma~\ref*{lem:#1}}}
\newcommand{\cor}[1]{\hyperref[cor:#1]{Corollary~\ref*{cor:#1}}}
\newcommand{\defn}[1]{\hyperref[def:#1]{Definition~\ref*{def:#1}}}
\newcommand{\tab}[1]{\hyperref[tab:#1]{Table~\ref*{tab:#1}}}
\newcommand{\alg}[1]{\hyperref[alg:#1]{Algorithm~\ref*{alg:#1}}}

\newcommand{\R}{\mathbb{R}}

\newcommand{\GB}{{\rm GB}}
\newcommand{\PAR}{{\rm PAR}}

\usepackage[bold]{hhtensor}
\usepackage{braket}
\usepackage{subcaption}
\input{Qcircuit}

\global\def \arxivmode {}
\ifx \arxivmode \undefined
\else
  
  \newcommand\arxivonly[1]{#1}
  \newcommand\prlonly[1]{}
\fi

\ifx \prlmode \undefined
\else
  
  \newcommand\arxivonly[1]{}
  \newcommand\prlonly[1]{#1}
  \renewcommand{\section}[1]{}
\fi

\begin{document}

\title{Quantum arithmetic and numerical analysis using Repeat-Until-Success circuits}
\author{Nathan Wiebe}
\author{Martin Roetteler}

\affiliation{Quantum Architectures and Computation Group, Microsoft Research, Redmond, WA (USA)}

\begin{abstract}
We develop a method for approximate synthesis of single--qubit rotations of the form $e^{-i f(\phi_1,\ldots,\phi_k)X}$ that is based on the Repeat-Until-Success (RUS) framework for quantum circuit synthesis.  We demonstrate how smooth computable functions $f$ can be synthesized from two basic primitives. This synthesis approach constitutes a manifestly quantum form of arithmetic  that differs greatly from the approaches commonly used in quantum algorithms.
The key advantage of our approach is that it requires far fewer qubits than existing approaches: as a case in point, we show that using as few as $3$ ancilla qubits, one can obtain RUS circuits for approximate multiplication and reciprocals.  We also analyze the costs of performing multiplication and inversion on a quantum computer using conventional approaches and find that they can require too many qubits to execute on a small quantum computer, unlike our approach.
\end{abstract}
\maketitle

\section{Introduction}
Classical arithmetic has a long history in quantum computation.  Shor's algorithm~\cite{shor:97}, linear systems algorithms~\cite{HHL09} and general purpose quantum simulation algorithms~\cite{BACS07} all deeply rely on  arithmetic functions, which are traditionally implemented using reversible circuits (see, e.g., \cite{Bennett:73,SBM:2006,MMD:2007,LJ:2014}). Beyond direct applications, there is a substantial body of literature that focuses specifically on implementing arithmetic functions in a reversible fashion on a quantum computer, see e.g., \cite{VBE96,BCDP:96,DKR+06,CDKM:2004,TK05,bea02,vMI:2005}.  Unfortunately, the inability of reversible circuits to forget previous parts of the calculation carries with it a heavy price:  the number of qubits required can be large.  For example, the number of qubits required to use Newton's method to compute reciprocals in linear systems algorithms can easily stretch to several hundred qubits if ten or more bits of precision are required~\cite{CPP+13}.  This result is interesting because it suggests that the number of qubits required to implement certain inexpensive classical algorithms (such as Newton's method) may be far greater than the number of qubits needed for the remainder of the quantum algorithm.  Thus new methods for performing arithmetic and function synthesis may be needed to enable computationally useful examples of linear systems, or related algorithms, to be run on a small scale quantum computer.

We address this problem by introducing an alternative method for computing functions on quantum computers.  The idea behind our approach is to encode numbers in the amplitudes of a qubit, or more properly as polar angles on the Bloch sphere.  For example, we represent the number  $\phi$ as the quantum state  $e^{-i \phi X}\ket{0}$ where $X$ is the Pauli--X operator. 
Our approach, in effect, consumes copies of resource states prepared using operations drawn from a set of inputs inputs $\{e^{-i\phi_1 X},\ldots,e^{-i \phi_k X}\}$ to approximate a unitary $e^{-if(\phi_1,\ldots,\phi_k) X}$ for some smooth computable function $f:\mathbb{R}^k\mapsto \mathbb{R}$.  Here $f$ can be an elementary function such as multiplication or it can be a more complicated function like a trigonometric function or a reciprocal.  In this sense, existing circuit synthesis results~\cite{KMM13,WK13,PS13,BRS14,RS14} reduce to solving this problem for cases where $f$ is the constant function.  If we think of this task as arithmetic, rather than a unitary synthesis task, then outputting the result as a rotation may at first glance seem unnatural; however, several important algorithms including quantum linear systems and fitting algorithms~\cite{HHL09,CPP+13,WBL12,W14} require the result to be output in exactly this fashion.  

Computing arithmetic in amplitudes circumvents the use of a qubit representation for the output and removes many of the ancillas needed in the computation.  This results in a substantial reduction in the number of qubits used relative to conventional approaches.  A further advantage of our approach is that $\phi$ need not be input as a qubit string: a quantum circuit that is promised to output $e^{-i \phi X}$ suffices.  A drawback is that it requires the use of amplitude estimation~\cite{BHM+00} or phase estimation when a qubit representation of $f(\phi_1,\ldots,\phi_k)$ is needed.

%The ability to approximate functions in this manner naturally  leads to questions that are variations of the basic common theme behind numerical analysis, namely how to best represent the function $f$, how to bound the roundoff and truncation errors in function evaluation, how to best circuitize its computation on a quantum computer, and how to do so in a manner that is as parsimonious as possible with respect to the number of resources it requires. Here we pay particular attention to the number of qubits needed since memory is a very expensive resource for existing and near--future quantum computers. 

We leverage the ``Repeat-Until-Success'' (RUS) paradigm of circuit design.  This paradigm broadly has two main features: (a) it allows probabilistic execution of unitaries, in particular the conditional application of unitary operations to parts of the quantum data depending on the outcomes of earlier measurements, and (b) all failures must be detectable and lead to an error that is correctable by a Clifford circuit.   In this sense, our approach is reminiscent of the KLM proposal for performing a CNOT gate in linear optical quantum computing~\cite{KLM01}.   See also~\fig{RUS} for a visualization of an RUS protocol.  The repeat until success moniker is thus earned because such algorithms do not fail but rather can be corrected and repeated until a successful outcome is obtained.

Most of the work on RUS circuits focuses on the  Clifford$+T$ gate set which is the universal gate set \cite{NC00} given by $\{H, S, T, {\rm CNOT}\}$, where $H$ denotes the Hadamard gate, $S={\rm diag}(1,i)$ is the phase gate, the $T=\sqrt{S}$, ${\rm CNOT}$ is the controlled--not gate, and we are allowed to apply these generators to any pair of qubits. The most costly gate in the gate set is assumed to be the $T$ gate because it is by far the most expensive operation to perform fault tolerantly in error correcting codes such as the surface code.  For these reasons, we also use the number of $T$--gates used, or $T$--count, as a metric to gauge the time--efficiency of our methods.

\begin{figure}[t]
\includegraphics[height=5cm]{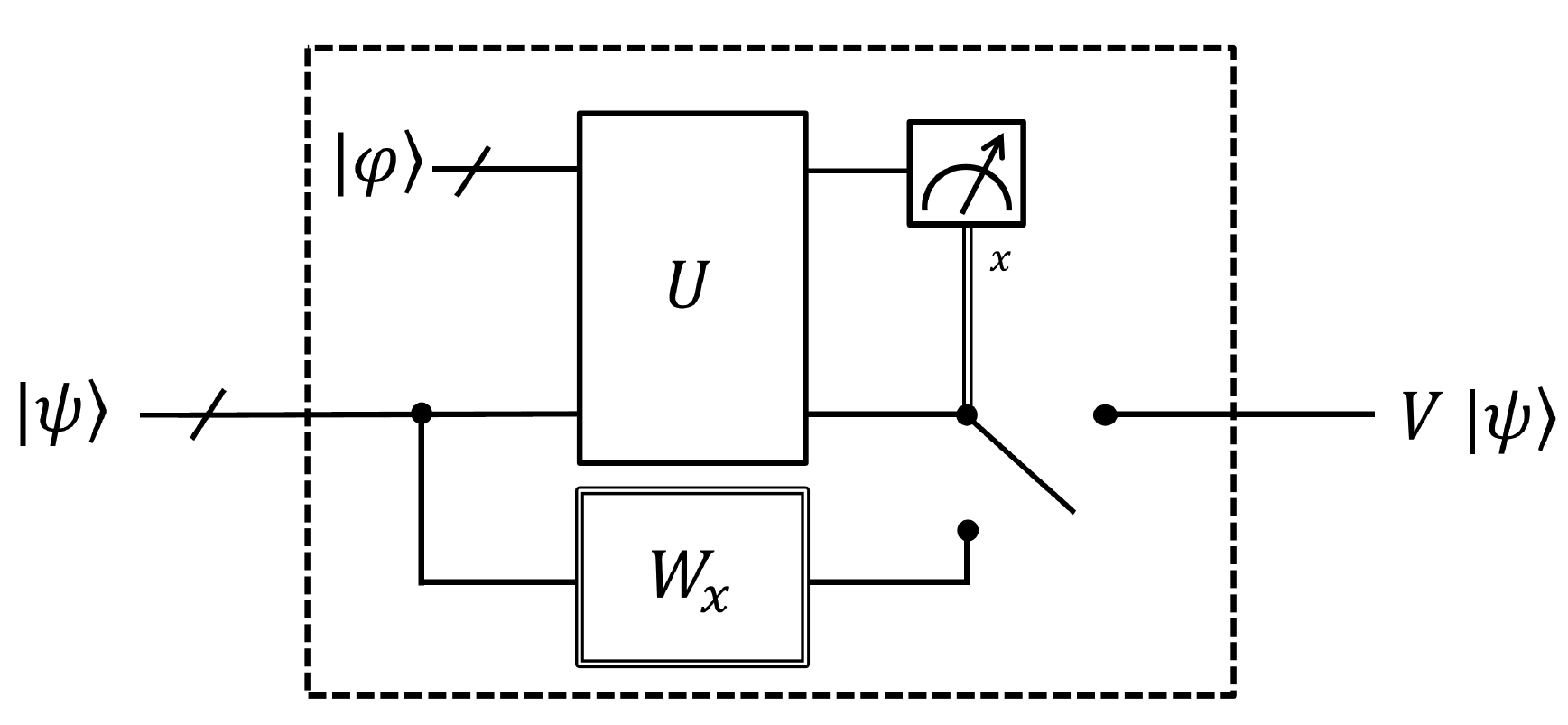}
\caption{\label{fig:RUS} Implementation of a unitary transformation $V$ via a Repeat-Until-Success (RUS) circuit protocol. The unitary transformation $U$ acts on the input state $\ket{\psi}$ and an ancilla state $\ket{\varphi}$ which has been prepared independent of the input. After application of $U$ the register holding the ancilla state is measured and in case the desired result $x$ is obtained---which, without loss of generality, can be chosen to be the outcome ``$(0,\ldots, 0)$''---the state $V \ket{\psi}$ is propagated to the output. Otherwise, a state is obtained that differs from the input $\ket{\psi}$ only by an application of a Clifford operations $W_x$ which itself might depend on the measurement outcome $x$. Hence we can ``route'' the output back to the input of $U$, together with a fresh copy of $\ket{\varphi}$. The routing is indicated by the switch that is controlled by $x$ being different from ``$(0,\ldots,0)$''. After $W_x$ has been applied, another attempt to compute $V \ket{\psi}$ can be started. The classical control of the Clifford gate $W_x$ is also indicated by double lines. The large dashed box indicates that this procedure is repeated until for the first time the desired outcome is measured, in which case the classically controlled ``switch'' is set in such a way that the output state $V\ket{\psi}$ can exist the RUS circuit.}
\end{figure}
The RUS paradigm that has gained interest after it was realized that it offers advantages for synthesizing single qubit unitaries \cite{WK13,PS13,BRS14}.
In particular: \cite{WK13} shows that this paradigm provides very low depth circuits for approximating small rotations and also elementary circuits that can approximately implement the square of the product of the input rotation angles. The circuits used there are RUS in the sense of \fig{RUS}, i.e., in the failure case they implement a Clifford gate that can be inexpensively corrected. In \cite{PS13,BRS14} RUS circuits are used for single qubit synthesis and the probabilistic nature of these circuits can be used to achieve an expected $T$-count of $1.15 \log_2(1/\varepsilon)$ for approximating a single qubit axial--rotation using the Clifford$+T$ gate set, up to an error of $\varepsilon$. This beats a lower bound on the average $T$--count of roughly $3\log_2(1/\varepsilon)$ that is known for the ancilla--free synthesis using Clifford + $T$ gates. 
%We adopt the Clifford$+T$ and furthermore use the RUS model to implement unitary transformations. The unitary transformations we are primarily interested in are approximations of rotations around the $X$-axis where the rotation angle is a classical function $f$. Arithmetic operations provide a particularly important class of such functions as they lie at the heart of several quantum algorithms~\cite{Shor:97,HHL09,CPP+13}. 

The repeat until success paradigm has several features that are especially valuable for quantum arithmetic.  First, they allow non--linear functions of the input angles to be computed.  This is very useful for arithmetic because multiplication is itself a non--linear function of its inputs.  Second, the correctability of the circuits allow them to be applied deterministically to an unknown quantum state (although the run time required to apply the transformation will vary).  This is significant because their success probability would shrink exponentially if such failures were not easily correctable.  Third, the transformations are irreversible.  This means that intermediate computations do not have to be retained for the entire calculation.  Finally, the qubit overheads of this form of arithmetic are extremely low since (in principle) the input, output and intermediate results can be stored in amplitudes rather than qubit strings.

%An important drawback of this approach is that the no--cloning theorem prevents replication of the results from these algorithms, such as $e^{-i \phi X}\ket{0}$.  This means that algorithms that rely on replicating the answers from prior computational steps (for example Newton's method) can be less time--efficient in this framework.  On the other hand, our approach to arithmetic works in a blackbox setting where a user only has access to a blackbox that performs $e^{-i\phi X}$ without requiring $\phi$ to be known as a bit string.

%It can actually be advantageous to perform arithmetic in this manner for a range of recent algorithms such as the linear systems algorithm~\cite{HHL09} as well as the quantum least squares fitting algorithm~\cite{WBL12,W14}.  One of the key tasks in these algorithms is to perform the transformation
%$$\sum_{x=1}^{2^n} a_x \ket{x}\ket{p_x}\ket{0} \mapsto \sum_{x=1}^{2^n} a_x \ket{x}\ket{p_x} e^{-if(p_x) X}\ket{0}.$$
%For the case of these algorithms, the function that is needed is either $f(p_x)=1/p_x$ or $f(p_x)=\sin^{-1}(1/p_x)$ depending on the desired trade-offs between accuracy and success probability.  We will show that our methods provide a remarkably qubit efficient method for computing functions such as these.  This potentially enables useful examples of quantum linear systems algorithms on much smaller devices than existing methods permit.

The paper is organized as follows.  We discuss the two circuits that form the core of our approximations in~\sec{circuits}.  We show that they can be used to provide an arbitrarily accurate approximation to $e^{-if(x)X}$ for a piecewise continuous function $f$ in~\sec{complete}.  We then apply the ideas in~\sec{complete} to implement multiplication of rotation angles in~\sec{multiplication} and then use both of these ideas to show how to implement reciprocals in~\sec{reciprocal}.  We discuss the use of caching strategies in~\sec{cache}, which provides a way to reduce the cost of recursive function evaluations using RUS circuits and also a way of storing the output as a qubit string.  We introduce an alternative method for implementing functions in~\sec{square} that we call square wave synthesis and apply it to computing reciprocals before discussing the problem of entangled inputs in~\sec{entangle} and then concluding.

\section{RUS Circuit Elements}\label{sec:circuits}
The core idea behind RUS arithmetic (or function synthesis) is to utilize measurement in clever ways to implement non--linear mappings between sets of input and output rotation angles.  This allows us to perform approximate arithmetic (or more general function synthesis tasks) in the rotation angles of qubits.  We need these circuits to have four properties.
\begin{enumerate}
\item The circuit performs a non--linear mapping that takes single qubit rotations whose angles are in $\mathbb{R}^k$ and maps these to a single qubit rotation whose angle is in $\mathbb{R}$ upon success.
\item The circuit will succeed with non--zero probability for all inputs.
\item The action of the circuit on the target state can be inexpensively corrected when the measurements fail.
\item The family of circuits considered must be able to exactly implement functions that scale as $O(\phi_1^{p_1}\ldots\phi_k^{p_k})$ for non-negative integers $p_k$ in the limit as $\max_k\{|\phi_k|\} \rightarrow 0$.
\end{enumerate}
The gearbox circuit~\cite{WK13} is a natural first guess for a class of circuits that satisfies these properties.
It is one of the earliest known classes of repeat until success circuits and has the following action upon success for input rotation angles $(\phi_1,\ldots,\phi_k)$
\begin{equation}
\GB~:~(\phi_1,\ldots,\phi_k)\mapsto \arctan(\tan^2(\arcsin(|\sin(\phi_1)|\cdots|\sin(\phi_k)|))).\label{eq:implicit}
\end{equation}
A diagram of the corresponding circuit is given in~\fig{GB} and the success probability of the gearbox circuits and their correction operations are given in~\tab{succprob}.
To be clear, the gearbox circuit (upon success) implements
\begin{equation}
\GB \left(e^{-i\phi_1 X}\ket{0}\cdots e^{-i\phi_k X}\ket{0} \ket{\psi}\right) = e^{-i \tan^{-1}(\tan^2(\arcsin(|\sin(\phi_1)|\cdots |\sin(\phi_k)|)))X}\ket{\psi}.\label{eq:explicit}
\end{equation}
Since we focus on computing the value of a function in the rotation angle of a state in this paper, we forgo using the more descriptive but cumbersome notation of~\eq{explicit} and instead use that of~\eq{implicit}.

\begin{figure}[t!]
\begin{minipage}{0.45\linewidth}
\[ \Qcircuit @R 1em @C 1.5em { 
\lstick{\ket 0}	&\gate{\phi_1}	&\ctrl{1}		&\gate{-\phi_1}	&\meter\\
		&			&\up{\vdots}  	&				&\\
\lstick{\ket 0}	&\gate{\phi_k}	&\ctrl{1}\qwx	&\gate{-\phi_k}		&\meter\\
		&\qw			&\gate{-iX}	&\qw					&\qw}
\]
\caption{Gearbox circuit with $k$ inputs\label{fig:GB}, where the gate $\phi_j$ denotes $e^{-i \phi_j X}$.  Success is achieved if~\emph{every} measurement reads $0$.}
\end{minipage}
\hspace{0.5cm}
\begin{minipage}{0.45\linewidth}
\[ \Qcircuit @R 1em @C 1.5em { 
\lstick{\ket 0}	&\gate{\phi_1}	&\ctrl{1}		&\multigate{2}{{\rm GHZ}^{-1}}	&\meter\\
		&			&\up{\vdots}  	&				&\\
\lstick{\ket 0}	&\gate{\phi_k}	&\ctrl{1}\qwx	&\ghost{{\rm GHZ}^{-1}}		&\meter\\
		&\qw			&\gate{(i)^{k-1}X}	&\qw					&\qw}
\]
\caption{Generalized $\PAR$ circuit with $k$ inputs\label{fig:PAR}, where ${\rm GHZ}^{-1}$ is the inverse GHZ measurement.  Success is achieved if~\emph{every} measurement reads $0$.}
\end{minipage}
\end{figure}

%The gearbox circuit has another intriguing property: whenever it fails it applies a single qubit Clifford gate, rather than the desired rotation, to the target qubit.  This means that any such errors can be trivially corrected and the circuit can be repeated until success is achieved, much like a Las--Vegas algorithm \cite{MR:95}.
Although the gearbox circuit satisfies the first three requirements, it does not satisfy the fourth because $\GB(\phi)$ does not have odd terms in its Taylor series expansion.  Odd terms can be introduced by shifting the input angles by $\pi/4$, but doing so reduces the success probability and also slows the convergence of the resulting Taylor series.  This problem can be addressed by introducing a second type of circuit that can inexpensively produce odd Taylor series.  We refer to this second design primitive as the ``generalized PAR circuit''.

\begin{table}[t!]
\[
\begin{tabular}{c@{\qquad}c@{\qquad}c@{\qquad}c}
\hline\\
Circuit & Function & Success probability& Correction circuit\\[1.5ex]
\hline\\
$\phi_1 \circ \phi_2$ & $\phi_1+\phi_2$ & $100\%$& --\\[1.5ex]
$\PAR(\phi_1,\ldots,\phi_k)$ & $\prod_{i=1}^k \phi_i + O(\max_i |\phi_i|^{k+2})$&
$\frac{1}{2}\left(\prod_{i=1}^k \cos(\phi_i)^2 + \prod_{i=1}^k \sin(\phi_i)^2\right)$ & $\openone\text{ or } 2\PAR(\phi_1,\ldots,\phi_k)$\\[1.5ex]
%$\frac{1}{2}(\cos(\phi_1)^2\cdots\cos(\phi_k)^2 + \sin(\phi_1)^2\cdots\sin(\phi_k)^2)$ & $\openone\text{ or } 2\PAR(\phi_1,\ldots,\phi_k)$\\
$\GB(\phi_1,\ldots,\phi_k)$ & $\prod_{i=1}^k \phi^2_i + O(\max_i |\phi_i|^{2k+2})$&
$\left(1-\prod_{i=1}^k \sin^2(\phi_i)\right)^2 + \prod_{i=1}^k \sin^4(\phi_i)$ &
%$(1-\sin^2(\phi_1)\cdots\sin^2(\phi_k))^2+ \sin^4(\phi_1)\cdots \sin^4(\phi_k)$ 
$e^{i\pi/4 X}$ (Clifford)\\[1ex]
\hline
\end{tabular}
\]
\caption{Circuits elements that we use to synthesize arbitrary functions: shown are the circuit names, the corresponding functions, the success probabilities, and the corresponding correction circuits.\label{tab:succprob}}
\end{table}

The generalized PAR circuit takes a very similar form to the programmable ancilla rotation circuit (or PAR circuit) proposed by Jones, Whitfield et al~\cite{JWM+12} for teleporting a pre--cached axial rotation that is stored in an ancilla qubit into a target qubit.  The principal differences between our circuits and those of~\cite{JWM+12} are: (a) our circuits are adapted to $X$--rotations and (b) they map multiple inputs to a single output rotation.  Since these circuits are clearly a generalization of the PAR concept, we refer to them collectively as~\PAR.  A diagram for these circuits is given in~\fig{PAR} and they perform the following rotation as a function of input angles:
\begin{equation}
\PAR~~:~ (\phi_1,\ldots,\phi_k) \mapsto \pm \arctan(\tan(\phi_1)\cdots\tan(\phi_k)).\label{eq:parmap}
\end{equation}
It is instructive to note that if $k=1$ then $\PAR~: \phi_1 \mapsto \pm \phi_1$ and thus it reduces to a $X$--rotation version of the PAR concept introduced in~\cite{JWM+12}.  

%A further interesting consequence of this circuit is that it actually can be interpreted as implementing Bayes' rule for a probability distribution over two models wherein $\phi_1$ encodes the prior distribution and $\phi_2,\ldots,\phi_k$ encode the ratios of the likelihood of events $1,\ldots,k-1$ occurring given that model $1$ is correct to the likelihoods of the data given model $2$ is correct.  This interpretation and its application to quantum inference will be elaborated on in greater detail in subsequent work.

The circuit for PAR (see~\fig{PAR}) is similar to the gearbox circuit but  uses a GHZ (generalized Bell state) measurement instead of one in the eigenbasis of $\ket{\phi_1}\ldots\ket{\phi_k}$.  The required GHZ measurement can be implemented for $n$ qubits by applying $n-1$ CNOT gates with control on qubit $1$ and each of the remaining qubits as a target and then applying a Hadamard gate to qubit $1$.  This reduces to a Bell state measurement when $n=2$.

The generalized PAR circuit satisfies three of the four required properties but fails the third criterion.  Eq.~\eq{parmap} shows that the third criteria cannot be satisfied because the circuit  either implements $e^{i \tan^{-1}( \tan(\phi_1) \cdots \tan(\phi_k)) X}$ or $\openone$ upon failure and the former error cannot be inexpensively corrected for arbitrary input states.  One approach for correcting such errors involves using the doubling down strategy of~\cite{JWM+12}. This strategy attempts to correct backwards--rotation errors by performing a rotation with twice the original rotation angle.  Should this succeed, the circuit will implement the desired target rotation.  Should it fail, then the correct rotation can be implemented by succeeding on a rotation with four times the initial rotation angle.  The mean number of attempts needed for this strategy to succeed is a constant, which means that it will be a viable error correction strategy in many contexts.  It is not viable here because the cost of doubling the rotation angle can grow rapidly with the size of the input rotation angles.  %Nonetheless, we will see that the generalized PAR circuit is often a repeat until success circuit since it is frequently applied to the state $\ket{0}$ and backwards rotations can be corrected by applying a Pauli--$Z$ operation in such cases.  
We show in~\lem{PAR} that oblivious amplitude amplification can be used to convert the generalized \PAR~circuit into a true repeat until success circuit that requires three times as many copies of $\phi_1,\ldots,\phi_k$ as the original circuit required.

These elements can then be used to implement an arbitrary function using the approach outlined in~\alg{apprx}.  The algorithm is entirely classical and its purpose is to iteratively construct a sequence of gearbox and PAR circuits that approximate the Taylor series of the target function $f$.  We show in~\sec{complete} that this procedure does not need to be iterated to infinite order to make the error arbitrarily small: it suffices to truncate at finite order.  Before showing this result, we first provide a rigorous discussion of the properties of the~\GB~and \PAR~circuits  and discuss how oblivious amplitude amplification can be used to remove the possibility of backwards rotation from~\PAR~circuits.

Note that we focus on $X$--rotations here, but our results can be trivially generalized to any Pauli operator by Clifford conjugating the rotation.  Even more generally, the $X$ operation can be replaced with any operator that is both Hermitian and unitary and all of our results still follow (although the correction operation may not necessarily be inexpensive for all such operations).  We focus on the case of single--qubit rotations for simplicity and also because of their ubiquity in quantum algorithms.
\begin{table}[t!]
%The gearbox and $PAR$ circuits are closely analogous to multiplication circuits, but they differ in a number of important ways.  A brief table of properties of $PAR$ circuits are given below.
\begin{align}
\PAR(\PAR(\phi_1,\ldots,\phi_k),\phi_{k+1})&=\PAR(\phi_1,\dots,\phi_{k+1}). &\text{(Associativity)}\\
\PAR(\phi_1,\phi_2)&=\PAR(\phi_2,\phi_1). &\text{(Commutivity)}\\
\PAR(\phi_1+m\pi,\phi_2+n\pi)&=\PAR(\phi_1,\phi_2). &\text{(Periodicity)}\\
\PAR(\phi_1 (-1)^{p_1}, \ldots, \phi_k (-1)^{p_k}) &=(-1)^{\sum p_k} \PAR(\phi_1, \ldots, \phi_k )&\text{(Oddness)}\\
\PAR(\phi_1,\ldots,\phi_k) &= \phi_1\cdots \phi_k + O(\max_k |\phi_k|)^{k+2}.&\text{(Non--linearity)}\\
\forall \phi_1 \in (-\pi/2,\pi/2), \lim_{k\rightarrow \infty} \PAR^{\circ k}(\phi_1,\phi_1) &= \frac{\pi}{4}(H_{\rm Heaviside}(\phi/\pi -1/4)+1). &\text{(Square Wave Limit)}
\end{align}
\caption{Table of properties of generalized PAR circuits. Here $H_{\rm Heaviside}$ is the Heaviside function.\label{tab:PAR}}
\end{table}

\subsection{Properties of Generalized PAR}\label{sec:genpar}
The programmable ancilla rotation (PAR) circuit is proposed in~\cite{JWM+12} as a means of pre--caching rotations in single--qubit ancilla states that can then be teleported onto a target qubit on demand.  These circuits useful in the context of quantum simulation because the rotations used in such simulations can be pre--cached at the beginning of the algorithm and then used as needed throughout the algorithm as needed without applying any non--Clifford operations.  Here we provide a generalization of this idea to multi--qubit inputs. We show that the resulting circuits perform a non--linear transformation to the pre--cached rotation angles that is analogous to multiplication.  Our specific claim is given in the following theorem:
\begin{theorem}
The operation $\PAR(\theta_1,\ldots,\theta_k)$ performs the operation $\exp(\pm i\arctan(\tan(\theta_1)\cdots \tan(\theta_k))X)$ with probability $\cos^2(\theta_1)\cdots \cos^2(\theta_k)+\sin^2(\theta_1)\cdots \sin^2(\theta_k)$ for $\theta_j\in \mathbb{R}$, and the positive and negative branches of the rotation occur with equal probability.  All other outcomes result in the identity operation being performed on $\ket{\psi}$.\label{thm:PAR}
\end{theorem}
\begin{proof}
The initial rotations and the multiply controlled $(i)^{k-1}X$ gate in the PAR circuit perform the following mapping:
\begin{align}
\ket{0}^{k} \ket{\psi} &\rightarrow (\exp(-i \theta_1X) \otimes \cdots \otimes \exp(-i\theta_k X) \ket{0}^{\otimes k} - \cos(\theta_1)\cdots \cos(\theta_k)\ket{0}^{\otimes k} - (-i)^k \sin(\theta_1)\cdots \sin(\theta_k)\ket{1}^{\otimes k})\ket{\psi}\nonumber\\
&\qquad+\ket{0}^{\otimes k} (\cos(\theta_1)\cdots \cos(\theta_k)) \ket{\psi} + \ket{1}^{\otimes k}(\sin(\theta_1)\cdots \sin(\theta_k)) (-iX)\ket{\psi}.\label{eq:PAR1}
\end{align}
Now let us define the following GHZ states:
\begin{align}
{\rm GHZ}\ket{0}\ket{0}^{k-1}:=\ket{{\rm GHZ}^+}&= (\ket{0}^k + \ket{1}^k)/\sqrt{2}\nonumber\\
{\rm GHZ}\ket{1}\ket{0}^{k-1}:=\ket{{\rm GHZ}^-}&= (\ket{0}^k - \ket{1}^k)/\sqrt{2}.
\end{align}
Expressing the result in terms of these GHZ states, we find that there exists a sub--normalized state $\ket{\phi}$ that is orthogonal to $\ket{{\rm GHZ}^+}$ and $\ket{{\rm GHZ}^-}$ such the right hand side of~\eq{PAR1} can be written as 
\begin{align}
\ket{\phi} &+ \frac{1}{\sqrt 2}\ket{{\rm GHZ}^+}(\cos(\theta_1)\cdots \cos(\theta_k) \openone -i \sin(\theta_1)\cdots \sin(\theta_k)X)\ket{\psi}\nonumber\\
&\qquad+ \frac{1}{\sqrt 2}\ket{{\rm GHZ}^-}(\cos(\theta_1)\cdots \cos(\theta_k) \openone +i \sin(\theta_1)\cdots \sin(\theta_k)X)\ket{\psi}.\label{eq:GHZeq}
\end{align}
If either $\ket{{\rm GHZ}^+}$ or $\ket{{\rm GHZ}^-}$ is observed then the outcome is 
\begin{equation}
\exp\left(\mp i\tan^{-1}\left(\frac{\sin(\theta_1)\cdots\sin(\theta_k)}{\cos(\theta_1)\cdots \cos(\theta_k)}\right) X\right)= \exp(\mp i \tan^{-1}(\tan(\theta_1)\cdots \tan(\theta_k))),
\end{equation}
and the probability of either measurement outcome occuring is 
\begin{equation}
P_{\pm} = \frac{\cos^2(\theta_1)\cdots\cos^2(\theta_k)+\sin^2(\theta_1)\cdots\sin^2(\theta_k)}{2}.
\end{equation}
Finally, since
\begin{equation}
\ket{\phi}:=(\exp(-i \theta_1X) \otimes \cdots \otimes \exp(-i\theta_k X) \ket{0}^k - \cos(\theta_1)\cdots \cos(\theta_k)\ket{0}^k - (-i)^k \sin(\theta_1)\cdots \sin(\theta_k)\ket{1}^k)\ket{\psi},
\end{equation}
it is clear that if a measurement outcome other than $\ket{{\rm GHZ}^+}$ or $\ket{{\rm GHZ}^-}$ is observed then the output state is $\ket{\psi}$ and hence no correction operation needs to be applied before repeating the circuit.
\end{proof}

Nonetheless, the PAR circuits do act as RUS circuits on particular inputs.  For example, if $\ket{\psi}=\ket{0}$ and in such cases the direction of the rotation can be switched by applying a $Z$--gate if necessary:
\begin{equation}
e^{-i (\arctan(\tan(\phi_1)\cdots\tan(\phi_k)))X}\ket{0}=Ze^{i\arctan(\tan(\phi_1)\cdots\tan(\phi_k))X}\ket{0}.
\end{equation}
Thus $\PAR$ can be repeated until success, at twice the success probability quoted above, if the generalized PAR gates are applied to a fresh ancilla.
We also show below that generalized PAR circuits (which also includes the PAR circuit~\cite{JWM+12}) can be converted into a repeat until success circuit for arbitrary inputs using oblivious amplitude amplification.
\begin{lemma}
The operation $\PAR(\theta_1,\ldots,\theta_k)$ can be converted to a repeat until success circuit with success probability $$\cos^2(\theta_1)\cdots\cos^2(\theta_k)+\sin^2(\theta_1)\cdots\sin^2(\theta_k)$$ that uses $3$ $\PAR(\theta_1,\ldots,\theta_k)$, and a constant sized Clifford circuit.  Upon failure, the correction operation is the identity gate.\label{lem:PAR}
\end{lemma}
A proof of this lemma is given in~\app{Lemma2}. Note that the basic case $\PAR(\theta_1)$ is entirely de--randomized by this procedure and all other cases are reduced to repeat until success circuits.  %Proof is given in the appendix.

These primitives can be composed by replacing the rotation used in the above circuits with an RUS sub-circuit.  For example, we can implement $\arctan(\tan^2(\phi_1)\tan(\phi_2))$ using the following circuit:
\begin{align}
\PAR(\GB(\phi_1),\phi_2)\equiv\qquad\qquad\qquad &\Qcircuit @R 1em @C 1.5em {
\lstick{\ket{0}}	&\gate{\phi_1}	&\ctrl{1}				&\gate{-\phi_1}	&\meter\\
\lstick{\ket{0}} 	&\qw 			&\gate{-iX}				&\qw			&\qw			&\ctrl{1}& \multigate{1}{{\rm GHZ}^{-1}}&\meter \\
\lstick{\ket{0}}	&\gate{\phi_2}	&\qw					&\qw			&\qw			&\ctrl{1}&\ghost{{\rm GHZ}^{-1}}&\meter \\
			&\qw			&\qw					&\qw			&\qw			&\gate{iX}	& \qw &\qw
}\label{eq:composedcirc}
\end{align}
The circuit is intended to be implemented by repeating the top--most measurement until a ``successful outcome'' of $0$ is measured.  Unsuccessful attempts can be corrected with a Clifford circuit.
Note that the output of the entire circuit is equivalent to $\PAR(\phi_1,\phi_1,\phi_2)$ upon success; however, the above method  typically requires fewer $T$ gates and has a lower online cost.  Equation \eq{composedcirc} can also be promoted to an RUS circuit using oblivious amplitude amplification since \GB~is an RUS circuit.

\begin{table}[t!]
\begin{align}
\GB(\phi_1)&=\PAR(\phi_1,\phi_1). & \text{(Equivalence)}\\
\GB(\phi_1,\phi_2)&=\GB(\phi_2,\phi_1).& \text{(Commutivity)}\\
\GB^{\circ k}(\phi)&= \arctan(\tan^{2^k} \phi_1). &\text{(Composition)}\\
\GB(\phi_1+m\pi,\phi_2+n\pi)&=\GB(\phi_1,\phi_2). &\text{(Periodicity)}\\
\GB((-1)^p\phi_1,(-1)^q\phi_2)&=\GB(\phi_1,\phi_2). &\text{(Evenness)}\\
\GB(\phi_1,\ldots,\phi_k) &= \phi_1^2\cdots \phi_k^2 + O(\max_k |\phi_k|)^{2k+2}.&\text{(Non--linearity)}\\
\forall \phi_1 \in (-\frac{\pi}{2},\frac{\pi}{2}) \lim_{k\rightarrow \infty} \GB^{\circ k}(\phi_1) & =\frac{\pi}{4}(H_{\rm Heaviside}(\phi/\pi -1/4)+1). &\text{(Square Wave Limit)}\\
\forall n,m\in\mathbb{Z}_+,~ \delta_{m,n}&=\frac{\int_{0}^\pi (\GB(2^mx)-\pi/4)(\GB(2^nx)-\pi/4) \mathrm{d} x}{\sqrt{\int_{0}^\pi (\GB(2^mx)-\pi/4)^2 \mathrm{d} x}\sqrt{\int_{0}^\pi (\GB(2^nx)-\pi/4)^2 \mathrm{d} x}}.& \text{(Orthogonality)}
\end{align}
\caption{Table of properties of Gearbox circuits.  $H_{\rm Heaviside}$ is the Heaviside function.~\label{tab:GB}}
\end{table}

\subsection{Properties of Gearbox Circuits}
The gearbox circuit is introduced in~\cite{WK13} as a means for expediently generating small rotation angles and rescaling the rotation angles output by circuit synthesis methods.
They earned their name because in the circuits can transform coarse input rotations into fine output rotations in analogy to a gear box.  
The following result, proven in~\cite{WK13}, provides justification for this claim as well as those made in~\tab{succprob}.
\begin{theorem}{[Wiebe, Kliuchnikov]}
Given that each measurement in $\GB^{\circ d}(\phi_1,\ldots,\phi_k)$ yields $0$, the circuit enacts the transformation $\ket{\psi}\mapsto e^{-iX\tan^{-1}(\tan^2(\theta))}\ket{\psi}$, where $\sin^2(\theta)= |\sin(\phi_1)|^2\cdots|\sin(\phi_k)|^{2}$.  This outcome occurs with probability $\cos^4(\theta)+\sin^4(\theta)$ and all other measurement outcomes result in the transformation $\ket{\psi}\rightarrow e^{i\pi X/4}\ket{\psi}$, regardless of the choice of $\phi_1,\ldots,\phi_k$.\label{thm:smallrot}
\end{theorem}
The $\GB$ circuit is therefore a repeat until success circuit, which means that the user can correct the result and try again  upon failure just like a Las--Vegas algorithm.  This provides a huge benefit here because it means that these circuit elements can be relied upon to produce the desired transformation.  The only downside, apart from low success probability near $\theta=\pi/4$, is that the rotations cannot be pre--cached into qubits and teleported into the system as per~\cite{JWM+12}.  Non--RUS~variants of \GB~that have this property are given in~\app{GB}.  These circuits are Clifford circuits and have the further advantage of lower online costs than $\GB$ circuits but cannot be reliably applied to an unknown quantum state.  Such circuits can also be used to simplify the state factory used in applications of floating point synthesis~\cite{WK13} that are optimized for low online $T$--counts.

\section{Completeness of Gearbox, PAR and Addition}\label{sec:complete}
\begin{algorithm*}[t!]
\caption{\label{alg:approximate} Taylor series based approximation algorithm.}
\rule{\linewidth}{1pt}
\begin{algorithmic}
\Require A smooth function $f\in C^M$, number of function variables $k$ and order of approximation $m \le M$.
\Ensure  A vector of functions $O$ such that  $|f(\phi_1,\ldots,\phi_k) - \sum_{\kappa =0}^k v[\kappa]| \in O(\phi^{k+1})$ where each $v[\kappa]$ can be implemented using gearbox and PAR circuits.

\vskip0.2em
\hrule
\vskip0.2em

\Function{generateApproximant}{$f$, $k$, $m$}
\State $F[0] \gets $ Taylor zeroth order Taylor series expansion of $f(\phi_1,\ldots,\phi_k)$ about $\phi_1 = \cdots = \phi_k =0$.
\State $v[0] \gets f[0]$ \inlinecomment{Store constant offset in first entry of output vector.}
\State $f \gets f - f[0]$
  \For{$i \in 1 \to m$}
	\State $F[i] \gets $ Lowest--order Taylor series expansion of $f(\phi_1,\ldots,\phi_k)$.
	\State $v[i] \gets $ Approximation to $F[i]$ using gearbox and PAR circuits correct to lowest order. 
	\State \inlinecomment{Approximation always exists but is not unique.}
	\State $f\gets f - v[i]$.
 \EndFor
  \State \Return $v$
  %\inlinecomment{We must normalize the updated weights before returning.}
\EndFunction
\end{algorithmic}
\rule{\linewidth}{1pt}\label{alg:apprx}
\end{algorithm*}

We saw in~\sec{circuits} that the \PAR~and \GB~circuits  have many properties in common with multiplication.  Indeed, the \PAR~circuit implements a normalized component-wise multiplication of two input vectors.  Since multiplication and addition can be used to approximately implement any continuous function on a compact domain, it is natural to expect that compositions of these circuit elements should also.  In order to see this formally, we need to introduce a lemma that shows that monomials can be approximated to within arbitrarily small error  using \GB~and \PAR~via time--slicing ideas reminiscent of those used in Trotter--Suzuki formulas.

\begin{lemma}[``time slicing lemma'']\label{lem:polynomial}
Let $|a|\le 1$, $b$ be a positive integer and $f(x)=ax^b$ be defined on a compact domain $U$. Then for any $\epsilon > 0$ the function $f$ can be implemented as an RUS circuit using \GB, \PAR~and addition within error $\epsilon$, as measured the two--norm.
\end{lemma}
\begin{proof}
Assume $b$ is even and $0<a\le1$.  Then there exists $k$ such that $b=2k$.  From the non--linearity property of $\GB$ there exists a gearbox circuit with $k+1$ inputs such that
\begin{equation}
{\rm \GB}(\arcsin(\sqrt{a}),x,\ldots,x)=ax^b +O(x^{b+2}).\label{eq:leadingapprox}
\end{equation}
Now assume than $-1\le a <0$.  We can apply the same gearbox circuit with a controlled $iX$, rather than a controlled $-iX$ gate, to implement the appropriate negative rotation.  Thus we can approximate this function to the appropriate order if  $|a|\le 1$.

If $|a|>1$ then $\arcsin(\sqrt{a})$ is not well defined.  We can circumvent this problem by noting that $|a|/2^{\lceil \log_2 a \rceil} \le 1$
and $ax^b = 2^{\lceil \log_2 a \rceil} \left(\frac{a}{2^{\lceil \log_2 a \rceil}} \right)x^b$.
Thus it follows from~\eq{leadingapprox} that
\begin{equation}
2^{\lceil \log_2 a \rceil}{\rm \GB}\left(\arcsin\left(\sqrt{\frac{a}{2^{\lceil \log_2 a \rceil}}}\right),x,\ldots,x\right)=ax^b +O(x^{b+2}).
\end{equation}
Thus we can assume without loss of generality that $|a|\le 1$ since all other cases can be found by adding the results of several gearbox circuits with $|a|\le 1$.  This addition step can be performed by serial composition (i.e. running one circuit after another using the same output qubit).

The monomial $ax^b$ for odd $b$ can be approximated using \GB, PAR and addition in a similar manner.  Since $b$ is odd, we can always express $b=2k+1$ and hence the linearity property of PAR and~\eq{leadingapprox} then imply that for any $|a|\le 1$
\begin{equation}
{\rm PAR}(x, {\rm \GB}(\arcsin(\sqrt{a}),x,\ldots,x)) = ax^{b} + O(x)^{b+2}.\label{eq:pargb}
\end{equation}
%The result of~\eq{pargb} is an RUS circuit because the results of several PAR circuits are not summed.
Thus odd Taylor series can be formed by nesting $\PAR$ and $\GB$ circuits (such terms can also be generated directly using $\PAR$ circuits).

In order to use time slicing ideas to make the error less than $\epsilon$ for any $\epsilon>0$ we will need to convert~\eq{pargb} into an RUS circuit.
This can be achieved by using oblivious amplitude amplification to convert PAR into a repeat until success circuit as per~\lem{PAR}.  Once the circuits have been converted to RUS circuits we can write a large rotation as a sum of small rotations because repeat until success circuits can be reliably corrected when errors occur and addition can be performed deterministically using serial composition (i.e. $e^{-iaX}e^{-ibX} =e^{-i(a+b)X}$).  Since the error in small rotations shrinks faster than the size of the rotation, this process allows arbitrarily small errors to be achieved in exact analogy to the use of Trotter--Suzuki algorithms in quantum simulation~\cite{BACS07}.  To see that slicing can make the error arbitrarily small, let $r>0$ be a parameter that is the $b^{\rm th}$ root of the number of slices used and note that (for odd $b$)
\begin{align}
r^{b}{\rm PAR}(x/r, {\rm \GB}(\arcsin(\sqrt{a}),x/r,\ldots,x/r)) &= r^{b}\left(a(x/r)^{b} + O(x/r)^{b+2}\right)\nonumber\\
& = a x^{b} + O(x^{b+2}/r^2).
\end{align}
Thus for any $x$, the error can be made less than $\epsilon$ for any $\epsilon>0$ by choosing $r\in \Theta(x^{{1+b/2}}/\sqrt{\epsilon})$.  Because $x$ is taken from a compact domain, it is possible to maximize over the values of $r$ to find a value that uniformly makes the error at most $\epsilon$.
The case for even $b$ is similar except since \GB~gives a repeat until success circuit oblivious amplitude amplification is not needed in such cases.  This implies that the $L_2$ distance between $ax^b$ and its approximant can be made less than $\epsilon$ for all $\epsilon>0$  under these assumptions.
\end{proof}

\begin{theorem}\label{thm:analthm}
Any function that is analytic on a compact domain $U\subset \mathbb{R}$ can be implemented within error $\epsilon$, as measured by the two--norm, using \PAR~and \GB~and addition as an RUS circuit.
\end{theorem}
\begin{proof}
Let $f$ be analytic on $\mathbb{R}$ then by definition there exists, for any $\epsilon>0$ and $x$, $\{a_j\}$ and $p$ such that
\begin{equation}
\left|f(x) - \sum_{j=0}^p a_j x^j\right|\le \epsilon/2.\label{eq:eps1}
\end{equation}
\lem{polynomial} shows that for any $\epsilon>0$ an RUS circuit can be constructed that approximates $\tilde f_j(x)$ using \GB, PAR and addition such that
\begin{equation}
|\tilde f_j(x) -  a_j x^j|<\epsilon/(2p).\label{eq:eps2}
\end{equation}
Then two uses of the triangle inequality and~\eq{eps1} and~\eq{eps2} gives us
\begin{equation}
\left|\sum_{j=0}^p \tilde f_j(x) - f(x)\right|\le \sum_{j=0}^p\left| \tilde f_j(x) - a_j x^j\right|+ \left|f(x) - \sum_{j=0}^p a_j x^j\right|\le \epsilon.
\end{equation}
\end{proof}
A natural generalization of this theorem is given below
\begin{corollary}
Any piecewise continuous function $f$ on a compact domain $U\subset\mathbb{R}$ can be approximated to within arbitrarily small error in the $2$--norm using an RUS circuit formed using \PAR, \GB~and addition.
\end{corollary}
\begin{proof}
Proof trivially follows from~\thm{analthm} and the Weierstrass approximation theorem.
\end{proof}

Any function that is smooth and computable can therefore be synthesized using repeat until success circuits, but one important point that has not been discussed yet is how many repetitions of these circuits are needed before a success is observed with high probability.  Such estimates are not necessarily easy to find because the execution and error correction steps that are needed in a tree--like repeat until success circuit, such as $\GB^{\circ k}(\phi)$, are difficult to compute.  Recursive expressions for the mean and variance of the number of rotations needed for a composed RUS circuit are provided in~\app{meanvar}.  The mean and the variance allow one to use Chebyshev's inequality to construct a confidence interval for the number of rotations required before a successful result is observed.  We use these results below where we apply the ideas in this section to implement multiplication and reciprocals.  We refer to arithmetic performed in this fashion as RUS arithmetic.

\section{Multiplication}\label{sec:multiplication}

%The ideas in the proof of~\thm{analthm} can be used to implement some important functions.  In the following section we will provide a number of examples of functions that can be implemented using these ideas, starting with multiplication circuits in the present section. 

Multiplication is perhaps the most important application for RUS function synthesis  because of its ubiquity in both quantum algorithms and numerical approximations to other functions such as reciprocals.
The key result that we will show in this section is that repeat until success circuits can be used to implement a form of multiplication that requires a constant number of ancilla qubits.  In contrast, most methods that have been proposed thus far for implementing multiplication require a number of qubits that scales at least logarithmically in the number of bits of precision needed.
As per the ideas of the previous section, we do not calculate the value of the product into a qubit string but instead we provide methods for approximately implementing
\begin{equation}
\ket{\phi_1}\ket{\phi_2}\ket{0} \mapsto \ket{\phi_1}\ket{\phi_2}e^{-i \phi_1 \phi_2 X}\ket{0}.
\end{equation}
It is worth noting that in some cases, such as linear systems algorithms or general purpose quantum simulation algorithms, it is actually desirable to have the result applied directly as a rotation.  Thus our form of RUS arithmetic directly outputs the desired rotation, rather than outputting a qubit string that then can be used to implement the rotation using a sequence of controlled rotation gates.

For simplicity, we will first assume that the user can perform a set of rotations, provided as black boxes, that can implement the single--qubit rotations $\{e^{-i \phi_j X}: j=1,\ldots, k\}$ and also has access to a circuit based quantum computer.  Later, we consider specific applications we estimate the total number of $T$--gates  required by the algorithm when these oracles are replaced by Clifford + $T$ circuits.

%The goal is then to implement $e^{-i \prod_{j=1}^k \phi_j X}$ within some fixed error tolerance (where error is measured with respect to some appropriate norm).   Since we encode numbers as rotation angles, this is equivalent to computing an arbitrary function in the rotation angle of a qubit. Also, as $\phi_1,\ldots,\phi_j$ are not necessarily known to the user, we require that this function work for an appropriate range of inputs rather than the specific ones given.  

The $\PAR$ circuit provides simplest approximation to multiplication within our framework but is only accurate to $O(\phi^4)$.  The relative error will be minuscule for cases where the input angles are small; however, if $\phi\approx 0.5$ then the relative error in the circuit is $16\%$.  This means that more accurate multiplication formulas will be needed for multiplying modestly large numbers if the error tolerance is small.  The following lemma shows that a high--order multiplication formula can be constructed that utilizes a small number of qubits.
\begin{lemma}
Let $\max\{|\phi_1|,|\phi_2|\}=x$, then $e^{-i \phi_1 \phi_2 X}$ can be approximated to within error $O(x^{q+2})$  for any even integer $q\ge 2$ on a compact domain $U\subset \mathbb{R}$. At most $O(\log q)$ qubits are needed to perform the multiplication.\label{lem:costmult}
\end{lemma}

\begin{proof}
Our proof proceeds inductively.  \lem{polynomial} shows that an approximate multiplication circuit $M_4(\phi_1,\phi_2)$ can be performed to within error $O(x^4)$.   We then see from Taylor's theorem that
\begin{equation}
\phi_1\phi_2 - M_4(\phi_1,\phi_2) = \phi_1\phi_2\sum_{j'+j=2} c_{j,j'} \phi_1^{j} \phi_2^{j'} + O(x^6),\label{eq:indhyp}
\end{equation}
where $c_{j,j'}\in \R$ are some constants. In~\eq{indhyp} we use the property that the error is of order $O(x^6)$ rather than $O(x^5)$ as $M_4(x)$ is
an even power series in $\phi_1$ and $\phi_2$.

Now, let us assume that for some $q\ge 4$ and some $c_{j,j'}$
\begin{equation}
\phi_1\phi_2 - M_q(\phi_1,\phi_2) = \phi_1\phi_2\sum_{j+j'=q-2} c_{j,j'} \phi_1^{j} \phi_2^{j'} + O(x^{q+2}).
\end{equation}
\lem{polynomial} shows that each term in the above power series can be implemented within error $O(x^{q+2})$.  In particular,
\begin{equation}
c_{j,j'}\phi_1^{j+1}\phi_2^{j'+1}=c_{j,j'} \times \PAR(\phi_1,\phi_2,\GB(\phi_1^{j/2},\phi_2^{j'/2}))+O(x^{q+2}).
\end{equation}
It then follows that
\begin{align}
\phi_1\phi_2 - M_q(\phi_1,\phi_2)-\sum_{j+j'=q-2}c_{j,j'}\times \PAR(\phi_1,\phi_2,\GB(\phi_1^{j/2},\phi_2^{j'/2}))&=O(x^{q+2}).
\end{align}
This process can be iterated in order to implement each $\phi_1^j$ and $\phi_2^{j'}$ using $\GB$ and $\PAR$.  
Now let us define $M_{q+2}(\phi_1,\phi_2)$ to be the approximant formed in this manner.  Note that because $\GB$ is an even function and $\PAR$ is an odd function, we have that there exist $c'_{j,j'}$ such that
\begin{align}
\phi_1\phi_2 - M_{q+2}(\phi_1,\phi_2)&=\phi_1\phi_2\sum_{j+j'=q} c'_{j,j'} \phi_1^{j} \phi_2^{j'}+O(x^{q+4}).
\end{align}
This demonstrates the induction step of our proof.  Now using~\eq{indhyp} as our induction hypothesis, we arrive at the conclusion that we can approximate multiplication to within error $O(x^{q+2})$ for any $q\ge 2$.

The composed gearbox circuit $\GB^{\circ k}(\phi)$ yields a rotation angle of $\phi^{2^k}+O(\phi^{2^k+2})$ using $k$ recursive applications of the circuit. Each application of the circuit requires $1$ additional qubit.  Therefore it is trivial to see by induction that $\phi^{2^k}$ can be approximated using at most $k+1$ qubits since $\phi^2$ requires $2$ qubits.
If we consider a binary expansion of $j$ then we see that at most $\lceil \log_2 (q/2-1)\rceil$ bits are needed to encode $j/2$.  Since the qubits used each of the $\phi^{2^k}$ in the decomposition $\phi^{j/2}=\phi^2\phi^4\cdots$ can be recycled and used to implement the other $\phi^{2^{k'}}$ terms, the number of qubits required to perform each of these terms is at most the number of qubits required for the most memory intensive calculation. This means $\lceil \log_2 (q/2-1)\rceil$ qubits will suffice, if we exclude the output qubits.  Since there are at most $\lceil \log_2 (q/2-1)\rceil$ terms in the decomposition, the total number of output qubits is $O(\log q)$ as well.  An additional $5$ qubits are needed to store the remaining arguments to the function and implement the multiply controlled Toffoli gate inside the gearbox circuit.  Thus the space total space requirements for this circuit scale as $O(\log q)$.
 
%The total number of additional qubits needed for this procedure (if we do not demand the circuit is executed as an RUS circuit and include an extra ancilla to implement $c_{j,j'}$) is at most $q/2+3$.  We use an extra qubit to convert this into an RUS circuit, as demonstrated by~\lem{PAR}, and thus the total number of ancillas needed is at most $q/2+5$ if the target qubit is also considered. 

\end{proof}

\lem{costmult} gives a procedure that can be used to construct a multiplication formula that has error that has arbitrarily high order error scaling using a number of qubits that scales logarithmically with the order of the multiplication formula.  This allows us to trade off space usage and time--complexity for any multiplication because the number of time--slices needed to achieve error $\epsilon$ for any such multiplier scales as $O(1/\epsilon^{2/q})$ for any fixed $q\ge 4$.  This suggests that if $c_{j,j'}\le 1$ for all $q$ then sub--polynomial scaling with $\epsilon$ can be achieved by choosing $q$ to be a function of $\epsilon$.  
 We leave rigorously demonstrating sub--polynomial scaling for higher--order multiplication formulas as an open problem. 

An important remaining issue is that many angles that appear naturally in problems, such as reciprocal calculation using the binomial method, will be approximately $1$.  These rotations can be implemented using time slicing as per~\lem{polynomial} but the cost of doing so may be prohibitive.  Instead, it makes sense to use a Taylor series expansion centered around $x=1$ rather than $x=0$.  We formally state this in the following corollary.
\begin{corollary}\label{cor:bigangle}
Assume that $\phi_1\approx \phi_2 \approx 1$ then $\phi_1 \phi_2$ can be implemented within error $O(\max\{|1-\phi_1|,|1-\phi_2| \}^{q+2})$ using at most $O(\log q)$ qubits.  Similarly, if $\phi_1\approx 0$ and $\phi_2 \approx 1$ then $\phi_1 \phi_2$ can be implemented within error $O(\max\{|\phi_1|,|1-\phi_2| \}^{q+2})$ also using at most $O(\log q)$ qubits.
\end{corollary}
\begin{proof}
Assume that $\phi_1\approx \phi_2 \approx 1$ then from~\lem{costmult} we have that
\begin{align}
\phi_1 \phi_2 &= -1 +\phi_1 +\phi_2 + (1-\phi_1)(1-\phi_2)\nonumber,\\
&= -1 +\phi_1 +\phi_2 + M_q(1-\phi_1,1-\phi_2)+O(\max\{|1-\phi_1|,|1-\phi_2|\}^{q+2}),
\end{align}
and this operation can clearly be implemented using at most $O(\log q)$ qubits.

\begin{table}[t!]
\begin{tabular}{c@{\quad}c@{\quad}c@{\quad}c}
\hline\\
Name & Multiplication formula & Error & Qubits (RUS)\\[1.5ex]
\hline\\
$M_4$ & $\PAR(\phi_1,\phi_2)$ & $O(x^4)$ &$4$\\[1.5ex]
$M_6$ & $\PAR(\phi_1,\phi_2,\frac{\pi}{4}- \GB(\gamma_2,\phi_1)-\GB(\gamma_2,\phi_2))$ & $O(x^6)$&$5$\\[1.5ex]
$M_8$ & $\PAR(\phi_1,\phi_2,\frac{\pi}{4}- \GB(\gamma_2,\phi_1)-\GB(\gamma_2,\phi_2),\frac{\pi}{4}-\GB(\gamma_3,\phi_1)-\GB(\gamma_3,\phi_2)+\GB(\gamma_2,\phi_1,\phi_2))$ & $O(x^8)$&$7$\\[1ex]
\hline
\end{tabular}
\caption{Lowest three orders of multiplication formula designed using the method of~\lem{costmult} but optimized for execution as an RUS circuit when acting upon $\ket{0}$. We use the constants $\gamma_2 = \arcsin(1/\sqrt{6})$ and $\gamma_3=\arcsin(1/\sqrt{15})$. Circuits are optimized for width and are meant to be executed from right to left in the \PAR~to allow the left most qubits to be used as ancillas for the \GB~operations appearing to their right.\label{tab:multform}}
\end{table}

\begin{table}[t]
\[
\begin{tabular}{c@{\qquad}c@{\qquad}c@{\qquad}c@{\qquad}c}
\hline\\
$x$ & $|x^2|$& $|M_4(x,x)-x^2|$&$|M_6(x,x)-x^2|$&$|M_{8}(x,x)-x^2|$\\[1.5ex]
\hline\\
$0.01$ &$1\times 10^{-4}$& $6.7\times 10^{-9}$ &$6.6\times 10^{-14}$ & $5.5\times 10^{-17}$\\[1.5ex]
$0.05$ &$2.5\times 10^{-3}$& $4.2\times 10^{-7}$ &$1.0\times 10^{-9}$ & $2.2\times 10^{-11}$\\[1.5ex]
$0.10$ &$1\times 10^{-2}$& $6.7\times 10^{-5}$ &$6.6\times 10^{-8}$ & $5.5\times 10^{-9}$\\[1.5ex]
$0.5$ &$2.5\times 10^{-1}$& $4.0\times 10^{-2}$ &$9.4\times 10^{-4}$ & $1.9\times 10^{-3}$\\[1.5ex]
$1.0$&$1.0$ &$0.18$& $0.054$ &$0.088$\\[1ex]
\hline
\end{tabular}
\]
\caption{Errors in the first three orders of multiplication formulas as a function of input angles.\label{tab:multerror}}
\end{table}

Now let us assume that $\phi_1 \approx 0$ and $\phi_2 \approx 1$.  We can then use similar reasoning to show that
\begin{align}
\phi_1 \phi_2 &= \phi_1 - \phi_1(1-\phi_2)\nonumber,\\
&= \phi_1 - M_q(\phi_1,1-\phi_2)+O(\max\{|\phi_1|,|1-\phi_2|\}^{q+2}),
\end{align}
and again the resultant rotation can be implemented using $O(\log q)$ qubits.
\end{proof}

%These ideas can also be used to optimize the depth of the resultant circuits.  Assume that $M_q$ is the sum of $N$ terms that each must be synthesized using $\GB$ and $\PAR$ circuits.  We can compute each such term separately and then teleport the result using the $\PAR$ circuit (without amplitude amplification) with probability roughly $1/2^{N-1}$.  This accounts for the fact that if the teleported version of the rotation experiences $N$ errors then the result can be corrected by applying a $Z$ gate to the target.  Thus by preparing $O(1/2^{N-1})$ copies of the constituent rotations and non--deterministically teleporting them into the system, we can achieve the desired rotation with high success probability using $T$--depth that scales as the maximum $T$--depth of any of the $N$ terms.

\subsection{Sixth--order multiplication formulas}

As an example, we will show how to derive a sixth--order multiplication formula, $M_6$, from a fourth--order multiplication formula $M_4$, which we take to be the output of the PAR circuit.  If $x$ is a small parameter then we can evaluate the behavior of the function for two small inputs by examining the Taylor series of $\PAR(ax,bx)$.  By Taylor expanding the function in powers of $x$ (i.e. using $\arctan(x) = x -x^3/3 +\cdots$ and $\tan(x)=x+x^3/3+\cdots$ ) we find
\begin{equation}
\PAR(ax,bx) = \arctan(\tan(ax)\tan(bx)) = abx^2 + \frac{1}{3}\left(ab^3x^4 +ba^3x^4 \right) + O(x^6).\label{eq:PARerr}
\end{equation}
\eq{PARerr} shows that $M_4$ behaves as an ideal multiplication circuit but with $O(x^4)$ error.  These error terms could be canceled by using the fact that $\GB(ax)=a^2x^2 +O(x^4)$ and then applying two more $\PAR$ circuits in series and Taylor expanding the result:
\begin{equation}
\PAR(ax,bx) - \PAR(ax,bx,\GB(ax),\arctan(1/3))-\PAR(ax,bx,\GB(bx),\arctan(1/3)) = abx^2+O(x^6).\label{eq:PARexamp}
\end{equation}
This process can then be repeated to make the $O(x^6)$ terms zero and so forth.  

We do not use~\eq{PARexamp} in practice because it uses qubits too greedily and is the sum of three different $\PAR$ circuits.  Because $\PAR$ is only an RUS circuit when it acts on $\ket{0}$, the sum of three outputs from $\PAR$ is not an RUS circuit since the three rotations that constitute it are applied in series and hence cannot possibly all act on $\ket{0}$ in the limit of small $x$ (unless $ab=0$).  Thus oblivious amplitude estimation will be needed to convert two of the three generalized PAR circuits into repeat until success circuits.  \lem{PAR} shows that this involves roughly tripling the cost of the circuit and so it is desirable to optimize the circuits to avoid this when possible.  $M_6$ uses one particular strategy to address the problem.

\begin{figure}[t!]
\begin{minipage}{0.45\linewidth}
\includegraphics[width=\linewidth]{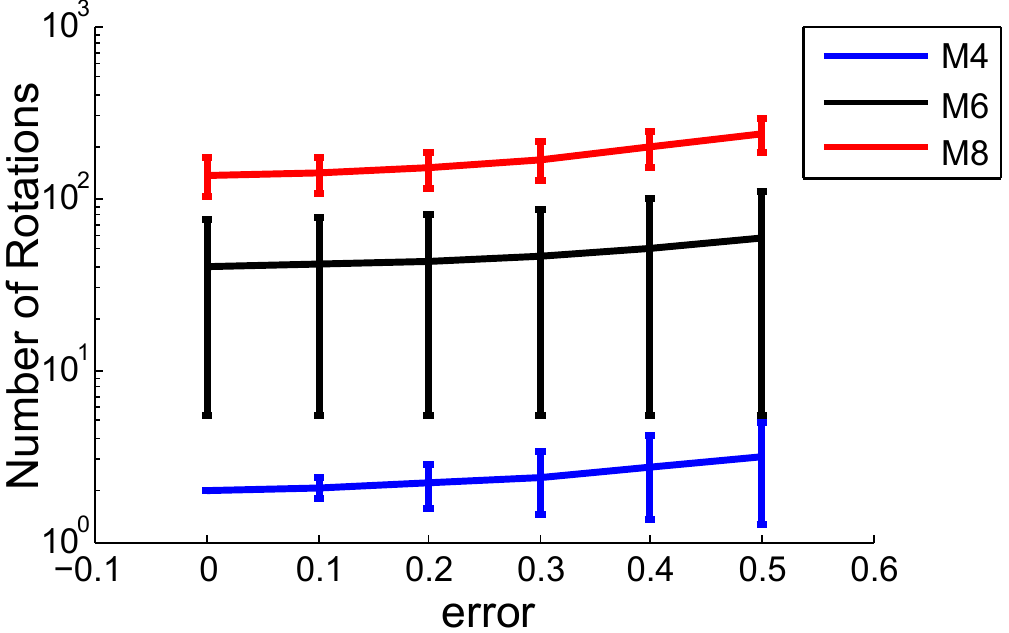}
\caption{Distribution for number of applications of $\phi_1$ and $\phi_2$ needed to approximate $e^{-i\phi_1\phi_2 X}$ using $M_4, M_6$ and $M_8$ for different values of $x=\max\{|\phi_1|,|\phi_2|\}$.\label{fig:costmult}}
\end{minipage}
\hspace{0.5mm}
\begin{minipage}{0.45\linewidth}
\includegraphics[width=\linewidth]{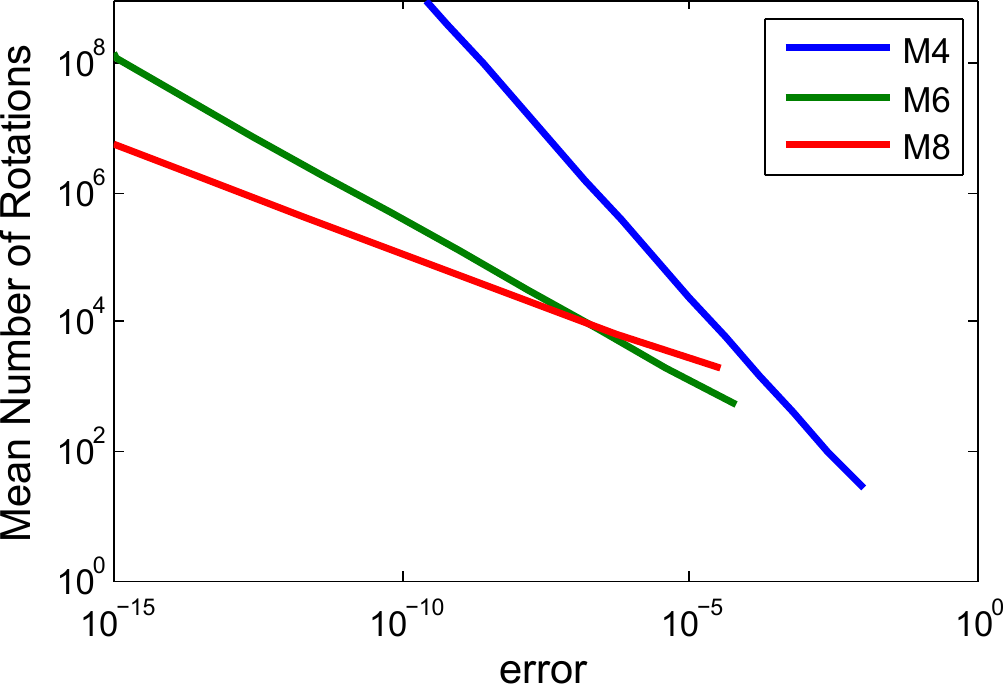}
\caption{Mean number of rotations (i.e. number of queries to an oracle that yields $e^{-i \theta X} \ket{0}$) needed to implement $e^{-i \theta^2 X}$ for $\theta=0.5$ using $M_4, M_6$ and $M_8$ versus approximation error for a number of time steps, $r^2$, ranging over $2^2, 2^4, \ldots, 2^{20}$.  \label{fig:spliterror}}
\end{minipage}
\end{figure}

The inspiration behind $M_6$  comes from noting that each term in~\eq{PARerr} consists of at least one $a$ and $b$.  This means that we can multiply $(1-a^2x^2/3-b^2x^2/3)$ by $\PAR(ax,bx)$ to achieve the desired result.  An efficient way to do this is to note that
\begin{equation}
\tan(x+\pi/4) = 1+2x +O(x^2).\label{eq:shift}
\end{equation}
Thus for any analytic function $f(x)$ we have that
\begin{align}
\arctan(\tan(ax)\tan(bx) \tan(f(x)+\pi/4))&= \arctan(\tan(ax)\tan(bx) + 2f(x)\tan(ax)\tan(bx) +O(f(x))^2).
\end{align}
The choice of $f(x)$ used in $M_6$ is $f(x) = -\GB(ax,\arcsin(\sqrt{1/6}))-\GB(bx,\arcsin(\sqrt{1/6}))= -a^2x^2/6-b^2x^2/6+O(x^4)$ which then, along with~\eq{PARerr}, gives us that
\begin{align}
\arctan(\tan(ax)\tan(bx) \tan(f(x)+\pi/4))&= \arctan(\tan(ax)\tan(bx) -a^3bx^4/3-b^3ax^4/3 +O(x^6)).\nonumber\\
&=\arctan(\tan(ax)\tan(bx)) -a^3bx^4/3-b^3ax^4/3 +O(x^6)\nonumber\\
&= abx^2+O(x^6).
\end{align}
$M_6$ therefore gives a sixth order approximation to the product of two numbers.  If only one time slice is needed to achieve the accuracy threshold for the problem then these circuits will be inexpensive; whereas if more than one time slice is needed then oblivious amplitude amplification will be needed to convert the $\PAR$ into an RUS circuit.  This causes the costs of these circuits to jump substantially during the transition from one to two slices.  

There are several ways in which the circuits could be optimized further for execution in cases where multiple time slices are needed.
One of the issues that arises stems from the fact that shifting the argument of these functions by $\pi/4$ (as per~\eq{shift}) comes at a steep price: it reduces the success probability of the circuit by roughly a factor of two.  This means that it should be used sparingly.  In cases where one slice is needed, its use results in an increase of a factor of $2$ in the expected cost of the circuit, which is superior (for small arguments) to the three--fold increase that would result from using amplitude amplification to allow the $\PAR$ circuits to be run in series as RUS circuits.  In cases where two or more slices will be needed, this trick becomes unnecessarily costly and the resulting circuits can be further optimized by opting instead for a strategy that is closer to~\eq{PARexamp}.

\subsection{Performance of Multiplication Circuits}

How well do these multipliers work for concrete inputs and concrete errors? We address this question by providing a table of elementary multiplication formulas for small rotation angles in~\tab{multform}.  We focus on formulas that are accurate for $x\approx 0$, but formulas adapted for $x\approx 1$ can be derived from these using the approach of~\cor{bigangle}.  The formulas $M_4, M_6$ and $M_8$ are highly accurate if $\max\{\phi_1,\phi_2\}\le 0.1$, but fail to accurately approximate multiplication for large rotations.  This is particularly noticeable with $M_8$, which is actually less accurate than $M_6$ for $\max\{\phi_1,\phi_2\}\ge 0.5$.  This is because the convergence of the Taylor series for $M_8$ is slowed due to the presence of large coefficients on the $O(x^8)$ terms that are introduced in this construction.

The cost of implementing the multiplication in terms of the number of times that the angles $\phi_1$ and $\phi_2$ need to be used is given in~\fig{costmult}.  There we see that the costs of implementing a multiplication using these formulas is minimal for small angles but increases with $x$ because of the costs incurred by the probability of failure in the gearbox and $\PAR$ circuits.  These costs are given in~\cor{cost}.  We ignore the costs of the Toffoli gates needed to perform these rotations because we assume that the cost of implementing the rotation will be substantially higher than that of the Toffoli gate.  

\fig{costmult} shows that $M_4$ costs approximately $2$ rotations, $M_6$ costs $40$ rotations and $M_8$ costs roughly $120$ rotations for $x\le 0.1$.  If time--slicing is used then all these circuits must be converted to genuine RUS circuits using~\lem{PAR}, which triples the cost of all rotations.  Regardless, we can use $r=5$ for $M_4$ at roughly the same cost as a single iteration of $M_6$ and $r=8$ for the same cost as an iteration of $M_8$.  The question remaining is, under what circumstances will using $M_4$ be preferable to using its higher order brethren. We see from~\fig{spliterror} that each of these methods for multiplying two rotation angles works best in a different regime.  $M_4$ is preferable for low accuracy rotations; whereas $M_6$ and then $M_8$ become methods of choice as the error tolerance shrinks.

\begin{table}[t!]
\begin{tabular}{c@{\qquad}c@{\quad}c@{\qquad}c@{\quad}c@{\qquad}c@{\quad}c@{\qquad}c@{\quad}c}
\hline\\
Multiplier & \multicolumn{2}{c}{$n=2$\phantom{111}} & \multicolumn{2}{c}{$n=4$\phantom{111}} & \multicolumn{2}{c}{$n=8$\phantom{111}} & \multicolumn{2}{c}{$n=16$} \\[0.5ex]
method     & $T$--count & qubits & $T$--count & qubits & $T$--count & qubits & $T$--count & qubits \\[1.5ex]
\hline\\
Carry-ripple & 2.34E+02 & 4 & 7.84E+02 & 8 & 2.80E+03 & 16 & 1.06E+04  &32\\[1.5ex]
Table-lookup & 3.38E+03 & 3 & 3.26E+06 & 3 & 3.98E+09 & 3 & 1.13E+13& 3\\[1.5ex]
$M_4$ &6.11E+01 &3 &1.97E+03 & 4 &4.64E+04 &4 &3.00E+07 &4 \\[1.5ex] 
$M_6$ &7.71E+02 &4 &1.67E+03 &4 &3.82E+03 &4 & 5.21E+05&5 \\[1.5ex] 
%$M_8$ & & & & & & & & \\[1ex] 
\hline
\end{tabular}
\caption{\label{tab:multiplier} A comparison of the resources required for space efficient multiplication on a quantum computer. Shown are circuit size (number of $T$-gates) and number of required qubits for two $n=2, 4, 8$ and $16$ bit numbers.  $M_8$ is not given because it is strictly more expensive than $M_6$ for this data set.  RUS synthesis was used to implement single--qubit rotations and a Toffoli construction that uses $7$ $T$--gates was used.  Extra qubits required for controlled Toffoli gates in $M_6$ are assumed to be recycled from prior steps.  All operations are assumed to be performed sequentially, costs for $M_4$ and $M_6$ fall substantially if parallel execution is permitted.}
\end{table}

\subsection{Comparison with classical methods for integer multiplication} \label{sec:compMult}

In this section we compare the resources required for RUS circuits for approximate multiplication with the traditional approaches of implementing multiplication by means of reversible circuits. Comparing classical approaches to implement a function $f(\phi_1, \ldots, \phi_m)$ of several inputs $\phi_i$, $i=1\,\ldots, m$---which are all assumed to be integers with the same precision, i.e., they are given by bit-strings of length $n$---with the ones described in this paper is not entirely straightforward: our methods assume that the inputs are given in form of rotations, whereas in classical approaches usually the inputs are given in a bit-string that encodes a basis state. To make the models comparable, we force the input and output types to be the same, i.e., a bit-string for the inputs and rotations for the outputs: for implementations based on classical circuits this means to encode the output from a reversible implementation $U_f$ of $f$ into a rotation. For this we use the circuit $Enc$ shown in \fig{enc}. Overall, we get a unitary circuit as shown in part (a) to compute the function by way of a classical reversible implementation and an RUS circuit as shon in part (b) to compute the function in Repeat-Until-Success style using measurements. In the remainder of this section we give estimates on the (expected) required resources for both cases (a) and (b), where we instantiate the function $f$ to be a multiplier of two $n$-bit integers $\phi_1$ and $\phi_2$. Later in \sec{reciprocal} we perform a similar analysis for the case where the function $f$ is the reciprocal function applied to $n$-bit integers. Out cost estimates for circuit size are based on the total number of $T$ gates used. Our cost estimates for required number of qubits do not include the qubits required to encode the inputs but only those required for everything else, i.e., the output qubits and any ancilla qubits that might be needed in the computation.

\subsubsection{Comparison with RUS methods}
We compare the two methods by choosing a problem for which both the inputs and outputs are well defined.  The problem that we use to benchmark these algorithms is one where two input numbers are provided as qubit strings: $\ket{\phi_1}\ket{\phi_2}$ and from these qubit strings we wish to implement the rotation $e^{-i \phi_1 \phi_2 X}$ within error at most $2^{-(n+1)}$, meaning that we have $n$--digits of precision in the output rotation.  We further constrain all algorithms to use gates only from the Clifford and $T$ library and take the cost to be the number of $T$--gates used.  The RUS synthesis method of~\cite{BRS14} is used to implement the rotations required in the inputs of $M_4$ and $M_6$ and the outputs of the carry-ripple and table-lookup multipliers.

A comparison between the resources required for the carry-ripple and the table-lookup multipliers with the methods $M_4$ and $M_6$ (which are given in~\sec{multiplication}) can be found in~\tab{multiplier}.  We see that  the number of qubits needed is substantially lower than those required to obtain comparable accuracy using the carry-ripple multiplier, although the $T$--count required for these implementations of RUS arithmetic are several times higher (except for $M_4$ at two bits of precision).  In contrast, lookup tables also require a constant number of qubits since but the $T$--count required by them is prohibitive for $n> 4$.  We see from this data that $M_4$ and $M_6$ give viable alternatives to performing multiplication using traditional methods on a fault tolerant quantum computer.  Perhaps most significantly, both methods require fewer than $5$ qubits to implement which means that they can be performed on existing quantum computers unlike carry-ripple multiplication.

As our RUS implementations are highly space efficient---recall that for instance $M_4$ requires only $4$ additional qubits to approximately compute the product of two $n$-bit numbers---we focus on classical implementations that optimize circuit width. Specifically, we consider the straightforward way of implementing a multiplication of two $n$ bit numbers using the na\"ive method of reducing the problem to $n$ additions. The advantage of this method is that it requires only $O(\log n)$ space in addition to the input and output registers. The size of the resulting circuit scales as $O(n^2)$. More advanced algorithms that achieve asymptotically better scaling in terms of total circuit size such as Karatsuba-Ofman or FFT-based methods seem to require a significant higher amount of space~\cite{K95}, so we do not consider them in the comparison. 

\begin{figure}[t]
\begin{tabular}{c@{\quad}c}
\raisebox{0.8cm}{\includegraphics[height=3.8cm]{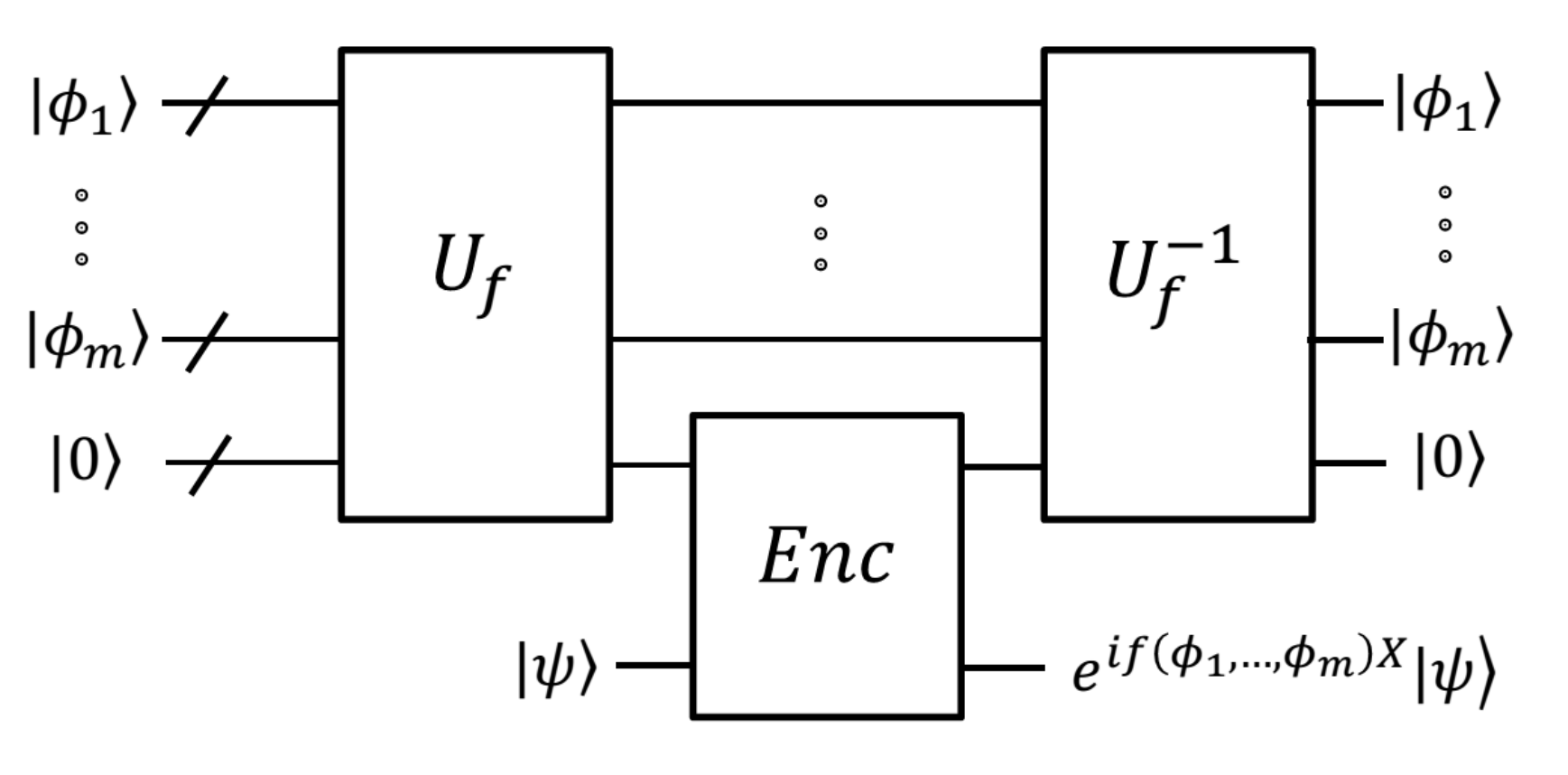}} & 
\includegraphics[height=5.8cm]{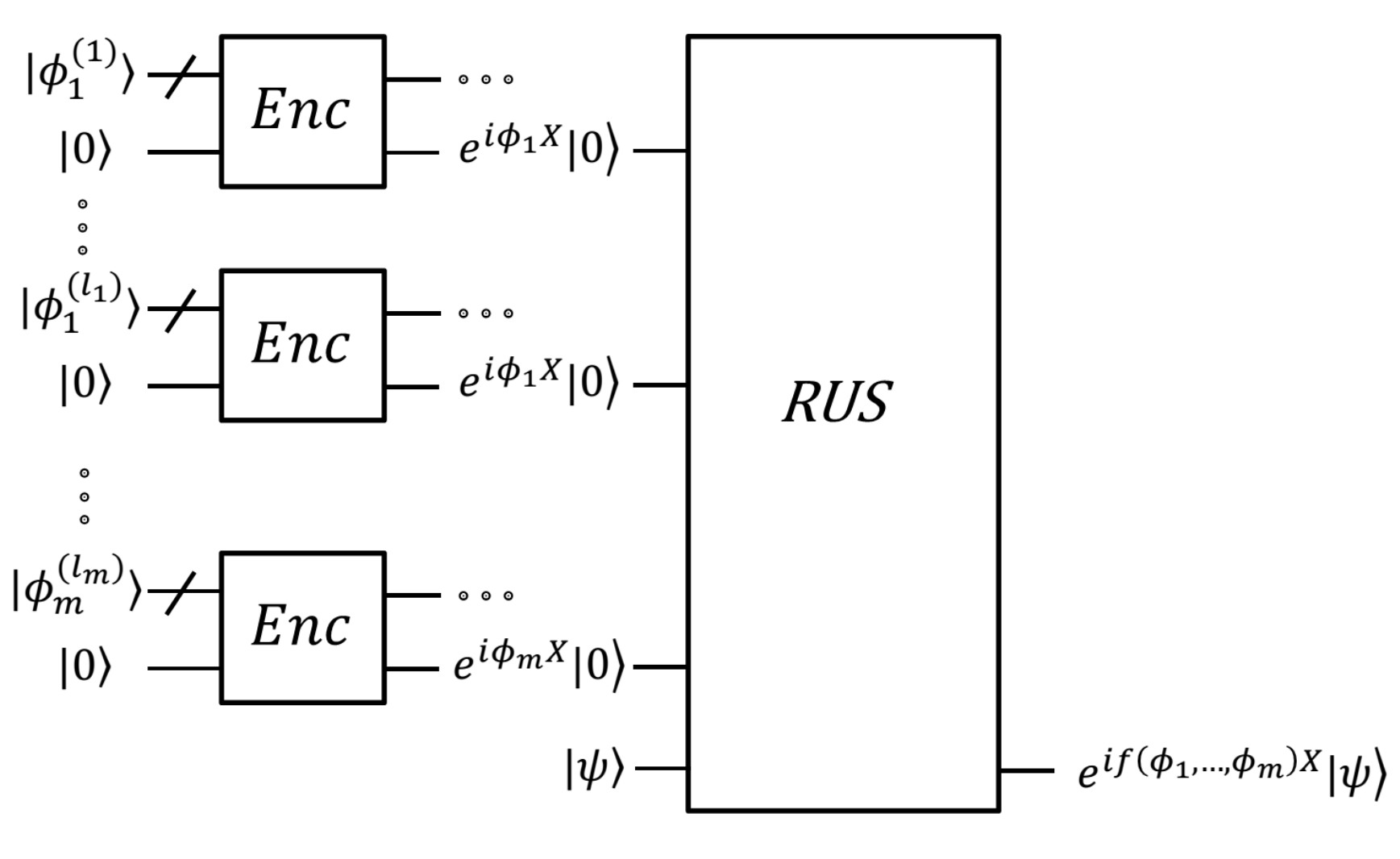}\\
(a) & (b)
\end{tabular}
\caption{\label{fig:comparison} Comparison of implementing a function $f$ as a rotation of a single target qubit $\ket{\phi}$. Shown are two implementations:  (a) via a classical, reversible circuit $U_f$ that first computes $f$ and then encodes the resulting bit-string $f(\phi_1, \ldots, \phi_m)$ as a rotation using a phase encoding circuit $Enc$. The implementation of $Enc$ in turn is shown in Figure \ref{fig:enc}. And (b) via an RUS circuit as in the methods presented in this paper. Note that in the RUS case, the input bit strings are encoded directly into rotatations which are then consumed by the RUS circuit. In contrast to the classical case, several rotations might be required to implement the target rotation: as shown in the figure $l_i\geq 1$ copies of the rotation corresponding to angle $\phi_i$, $i=1, \ldots, m$ are used, denoted by  $\ket{\phi_i^{(1)}},\ldots, \ket{\phi_i^{(l_i)}}$.  
}
\end{figure}

\begin{figure}[hbt]
\includegraphics[height=4.5cm]{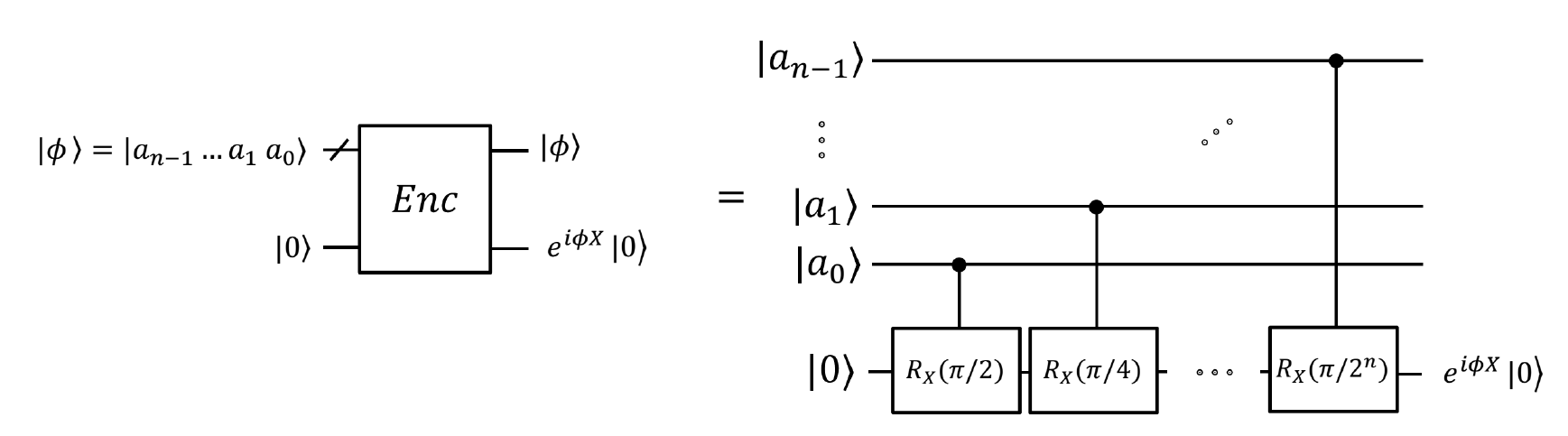}
\caption{\label{fig:enc} Implementation of the phase encoding $Enc$ of an $n$-bit string $\phi = a_{n-1}\ldots a_1 a_0$ using controlled rotations. The cost for each controlled rotation depends on the overall target error $\epsilon$ and scales as $O(\log 1/\epsilon)$ and is assumed to be the same for all rotations shown in the figure.}
\end{figure}

As our methods require only a constant amount of space, arguably the most meaningful comparison is to compare them with classical multipliers that only use a constant amount of memory to do so.
 While there is work on the complexity of computing output bits of the product in the context of space-bounded computation models such as OBDD or branching programs---e.g., on the middle bit \cite{WW:2007} or the most significant bit \cite{BK:2011}--we are unaware of work that addresses the space-bounded complexity for the {\em approximate} multiplication problem. In the absence of such results, we instead consider an extremal case of table-lookup computations, i.e., computations that can be implemented with a constant number of additional qubits but which implements the function in a brute force way that does not exploit any features that the function might possess. 

\subsubsection{Carry-ripple multipliers} The na\"ive method implements multiplication of two $n$-bit numbers $x=\sum_{i=0}^{n-1} x_i 2^i$ and $y=\sum_{i=0}^{n-1} y_i 2^i$ by performing $n-1$ additions where the summands are shifted versions of $x$. Overall, this requires scratch space that is logarithmic in the number of input bits. Moreover, carry-save techniques can be applied to keep circuit depth of the additions small \cite{PS:2013}, i.e., the first of the $n-2$ additions can be performed in constant depth using a suitable data structure. As we are interested in optimizing space, we chose a different path and consider a carry ripple adder \cite{DKR+06} which we then use $n-1$ times to produce the desired output. From the quantum circuits for additions that have been studied in the literature \cite{DKR+06,CDKM:2004,vMI:2005}, we pick one that has the property that a controlled adder can be implemented at relatively low overhead. 

While classically the na\"ive method needs only $\log n$ additional space, in the quantum case we have to store the intermediate results because the overall computation must be reversible. We obtain an upper bound of $3n$ qubits for the total space required, including input and output qubits as we can implement the multiplication using $n$ controlled adders that add shifted copies of the input $x$ to the output register, where the controls depend on the bits of $y$. The resultant $T$--count can be upper bounded by $n$ times the cost for an in-place adder that is controlled on a single qubit. 

To keep space as well as circuit size small, we choose the in-place adder presented in \cite[Figure 5]{DKR+06} as this has the particular useful feature that a {\em controlled} adder can be derived directly from it without using many additional control gates, an observation also used in \cite[Section V]{SM:2013}. Explicitly, by counting the number of gates we obtain that  
at most $12n$ Toffoli gates are required. Using the fact that a Toffoli gate can be implemented using $7$ $T$-gates \cite{NC00,AMMR:2013} we obtain an upper bound of $84n$ $T$-gates for a controlled adder. We now use $n$ controlled adders (of input size $n$, $n+1$, \ldots $2n$ bits) conditioned on the bits of $y$ for an overall $T$--count of $\sum_{k=n}^{2n} 12k = 18n^2+18n$  for the carry-ripple multiplier, i.e., for the implementation of $U_f$ as in \fig{comparison}. The cost for implementing $Enc$ can be estimated as follows: to be comparable with the RUS-based multipliers that produce $n+1$-bit approximations of the rotations, not all $2n-1$ bits of the output of $f$ have to participate in the controlled rotations as in \fig{enc}. Indeed, it is enough it the highest $n+2$-bits participate in order to get an $n+1$-bit approximation. As we then have $n+2$ rotations and we have the target error $\epsilon_{target} = 2^{-(n+2)}$, we obtain that we need at most $1.15 \log_2((n+2)/2^{-(n+2)})$ many $T$-gates per each rotation in $Enc$ where we choose to distribute errors uniformly, i.e., $\epsilon = \epsilon_{target}/(n+2)$ and we used the upper bound \cite{BRS14} for RUS-based single qubit decompositions. Putting everything together, we get an overall cost of $2 \cdot (18n^2+18n) + 1.15 (n+2) \log_2((n+2)/2^{(n+2)})$ many $T$-gates, where the leading factor of $2$ is due to the cleanup of the ancillas. For small values of $n$, the resulting upper bound estimates on the number of $T$-gates are shown in \tab{multiplier}. Note that the space bound on the number of qubits for the carry-ripple adders are given by $2n-1$ for storing the output of $f$ plus $1$ qubit for the finally resulting rotation.

\subsubsection{Table-lookup multipliers} The problem of computing the product of two $n$-bit numbers $x$ and $y$ can be restated as a problem of computing $2n-1$ Boolean functions $f_0(x_0, \ldots, x_{n-1}, y_0, \ldots, y_{n-1}), \ldots, f_{2n-2}(x_0, \ldots, x_{n-1}, y_0, \ldots, y_{n-1})$, i.e., one Boolean function $f_i : \{0,1\}^{2n} \rightarrow \{0,1\}$ for each output bit.  This function can be stored as a lookup table wherein the individual bits yielded by the Boolean functions are stored in an array.  We now consider the complexity of implementing all these functions via lookup tables. As we are interested in uniform families of circuits (as opposed to non-uniform circuit models in which lookup tables can be implemented in $O(n)$ time and $O(n)$ space) we are therefore looking for a quantum circuit that can implement a lookup-table with $2n$ inputs and $n+2$ outputs. Note that we need only the highest order $(n+2)$ of the result to be comparable to the RUS-based implementation, so we do not have to synthesize all $2n-1$ output functions. 
One simple way to upper bound the cost for implementing such a lookup-table is to assume that each output bit is implemented via a sequence of $2n$-fold controlled NOT gates $\Lambda_{2n}({\rm NOT})$, where $\Lambda_k(U)$ is defined as the operation $\Lambda_k(\ket{x_0, \ldots, x_{k-1}}\ket{\psi}) = \ket{x_0, \ldots, x_{k-1}}\ket{\psi}$ if $(x_0, \ldots, x_{k-1}) \not= (1, \ldots, 1)$ and $\Lambda_k(\ket{1,\ldots, 1}\ket{\psi}) = \ket{1, \ldots, 1} U \ket{\psi}$. 

From \cite{BBC+:95} follows that we can implement a $k$-fold controlled NOT (i.e., the case $U={\rm NOT}$) using at most $8k-24$ many Toffoli gates, provided that $k\geq 5$. For small values of $k$, a case analysis shows that $k=2$ requires $1$ Toffoli gate, $k=3$ requires at most $4$, and $k=4$ at most $10$. 

We now break these Toffoli circuits further down over the Clifford+T gate set, specifially, we count the number of $T$ gates. Using known implementations \cite{NC00,AMMR:2013} of the Toffoli gate over Clifford$+T$ it can be shown that its cost in terms of $T$-gates is given by $7$. Note, however, that often it is useful to consider a Toffoli up to a diagonal phase. It is known that this leads to savings in the $T$-count, specifially, an implementation of Toffoli up to a phase is known that requires 4 $T$ gates only \cite{BBC+:95,Selinger:13}. 

An analysis of the Toffoli network for the $k$-fold NOT given in \cite{BBC+:95}, while using as much as possible the cheaper Toffoli up to phase whenever phase cancellations are possible, reveals that for $k\geq 5$ the $T$-count can be upper bounded by $32k-84$. For small values of $k$, again a case analysis can be done which shows that $k=2$ requires at most $7$ $T$-gates, $k=3$ requires at most $22$ $T$-gates, and $k=4$ at most $52$. 

Each output bit is a Boolean function of $2n$ inputs and each non-zero line of the truth table is implemented by a $\Lambda_{2n}({\rm NOT})$ gate, up to Clifford gates (NOTs). The space overhead of this implementation is constant, namely we need $1$ additional qubit in order to implement the decomposition as in \cite{BBC+:95}. 
Note that as above we have to only implement the leading $n+2$ bits of the product. 
We make make the conservative worst case assumption that the Hamming weight of each output bit could potentially be as high as $2^{2n}$ and we have $n+2$ output bits, i.e., we get an upper bound of $(n+2) (2^{2n} (32 (2n)-84) = 2^{2n} (64n^2+44n-168)$  $T$-gates are required for computing the $(n+2)$ highest order bits of the product of two $n$-bit numbers using a reversible cicuit $U_f$. Here we used $1$ additional ancilla qubit to enable the linear time factorization of the $2n$-fold controlled NOT gates. 

As above we have an overhead of $1.15 (n+2) \log_2((n+2)/2^{(n+2)})$ $T$ gates to implement the $Enc$ gates. 
Putting everything together we get an overall cost of $2 \cdot 2^{2n}(64 n^2+44n-168) + 1.15 (n+2) \log_2((n+2)/2^{(n+2)})$ many $T$-gates, where the leading factor of $2$ is due to the cleanup of the ancillas. For small values of $n$, the resulting upper bound estimates on the number of $T$-gates are shown in \tab{multiplier}. Note that for the case $n=2$ we cannot use the formula as it falls within one of the special cases of small number of controls. In this case we obtain a bound of the $T$ gates arising from the reversible part of the circuit as $6,656$ gates to which the cost for $Enc$ has to be added.

Note that the space bound on the number of qubits for the table lookup implementation is given by $1$ qubit for the result of each output function, $1$ qubit as an ancilla, and $1$ qubit to store the final rotation, i.e., a total of $3$ qubits. 

It seems possible that this crude upper bound on the number of $T$ gates can be improved by reusing intermediate results and output bits \cite{MMD:2003,SSP:2013} or by applying synthesis techniques based on the Reed-Muller transform \cite{LJ:2014}, however, an exponential lower bound for $f_{n}$ (the ``middle bit'') of $\Omega(2^{n/2})$ is known for branching programs that are allowed to read the inputs a constant number of times \cite{WW:2007} and also for the most significant bit an exponential lower bound of $\Omega(2^{n/720})$ is known (for OBDDs which are a special case of branching programs) \cite{BK:2011}, hence even after optimization, the circuit complexity will be exponential in case there is only a constant amount of memory available. 

\section{Reciprocals}\label{sec:reciprocal}
An important gap in the application of the Harrow, Hassidim and Lloyd quantum algorithm for solving linear systems~\cite{HHL09} is the fact that a rotation of the form $e^{-i X/a}$ must be performed for some superposition over the values of $a$ stored in a quantum state.  The conventional approach to solving this problem is to provide a classical reversible circuit for the reciprocal and use it to compute $1/a$ into a qubit string stored in a tensor product with $a$.  With this value in hand, $e^{-i X/a}$ can be performed using a series of controlled rotations.  This procedure is discussed in detail in~\cite{CPP+13} and also in~\sec{newton}.  A major drawback of this approach is that many qubits are required to store $1/a$.   A method that goes directly from $\ket{a}$ to $\ket{a} e^{-i X/a}\ket{0}$ without needing to compute $1/a$ in a qubit register would therefore be quite useful.

Newton's method is perhaps the most commonly prescribed method for computing the reciprocal.  The reason for its popularity stems from the fact that (i) Newton's method converges quadratically for a good initial guess and (ii) the approach only requires multiplication and addition.  In particular, if $x_n$ is an approximation to the value of the reciprocal then Newton's method provides a new approximation
\begin{equation}
x_{n+1} = 2x_n - a x_n^2.\label{eq:newtrecur}
\end{equation}
This process begins with a reasonable guess for the value of the reciprocal, such as $x_1 = 2^{-\lceil\log_2 a \rceil}$, and~\eq{newtrecur} is then iterated until the error is sufficiently small.  Since the error shrinks quadratically, $n\in \Theta(\log\log 1/\epsilon)$ iterations suffice to reduce the error to at most $\epsilon$.
A direct application of Newton's method is not well suited for calculating the reciprocal using RUS arithmetic because each iteration requires four copies of $x_n$ and hence the total number of rotations required scales as $O(4^{n})$; making a direct application of this method costly.  This approach can be made more viable by unrolling the recurrence relation into a polynomial and then approximating the polynomial using the methods of~\sec{complete}, but the cost of implementing the resulting polynomial using RUS arithmetic can be prohibitive because the coefficients in the polynomial diverge exponentially.  Caching methods, described in~\sec{cache}, can also be used to reduce the cost of implementing Newton's method using RUS arithmetic at the price of increased circuit width.

We focus on two other methods for computing the reciprocal using RUS circuits.  First we discuss directly implementing a Chebyshev approximant to the reciprocal and then consider implementing the binomial method (also known as the IBM method) for implementing the reciprocal.  Both approaches yield practical methods for approximating $1/a$ using RUS arithmetic.
The first step in both of these methods involves rescaling $a$.
This step is important because it circumvents the problem of exponentially diverging coefficients that appears in a direct application of Newton's method.  This rescaling can be expressed as
\begin{equation}
\frac{1}{a} = 2^{-\lceil \log_2 a \rceil}\left(\frac{1}{2^{-\lceil \log_2 a \rceil} a }\right),
\end{equation}
where $2^{-\lceil \log_2 a \rceil} a \in [1/2,1]$.  We can therefore introduce a new variable
\begin{equation}
y= 1-2^{-\lceil \log_2 a \rceil}a,
\end{equation}
where $y\in [0,1/2]$.  In other words, we seek to find a power series approximation in powers of $y$ to
\begin{equation}
2^{-\lceil \log_2 a \rceil}\left(\frac{1}{1-y}\right).
\end{equation}
Three natural methods then arise for implementing this: Taylor series, Chebyshev polynomials and the binomial division algorithm.  Taylor series tend to provide poor accuracy for this application because of the slow convergence of the series near $y=1/2$ (as alluded to previously).  For this reason, we focus our attention on the remaining two methods.

\subsection{Chebyshev Polynomials}
Chebyshev polynomials are a complete set of orthogonal polynomials that can be used to represent any piecewise continuous function, such that the infinity norm  of the difference between the approximation and the original function is minimized.  Taylor series approximations (such as those used in~\sec{multiplication}) provide extremely accurate local approximations to a function but tend not to provide approximations that are accurate throughout the domain of the function.  Thus Chebyshev polynomials are often the preferred method for obtaining a polynomial approximation to a function on an interval.  The properties of these polynomials are well studied and discussed in detail in~\cite{AS12}.  
The key point behind this approach is that by doing a Chebyshev polynomial expansion, we can reduce the problem of finding the reciprocal to that of implementing a polynomial.  This can be achieved by using the multiplication formulas provided in~\sec{multiplication}.  The three lowest--order Chebyshev approximants to the rescaled reciprocal, $1/(1-y)$, are given in~\tab{cheb}.

\begin{table}[t!]
\begin{tabular}{c@{\qquad}c@{\qquad}c}
\hline\\
Gadget name & Formula & Maximum error\\[1.5ex]
\hline\\
$R_2$&$1.012194+.608948 y+2.664355 y^2$&$1.6\times 10^{-2}$\\[1.5ex]
$R_4$&$1.000359+.966359 y+1.490195 y^2-1.362554 y^3+5.019604 y^4$&$5.1\times 10^{-4}$\\[1.5ex]
$R_6$&$1.000012+.9980208 y+1.059785 y^2+.336629 y^3+4.386547 y^4-7.295458 y^5+9.456853 y^6$ & $1.2\times 10^{-5}$\\[1ex]
\hline
\end{tabular}
\caption{Chebyshev approximants to $1/(1-y)$ on $y=0\ldots 1/2$ where $y=1-2^{-\lceil\log_2 a \rceil}a$.\label{tab:cheb}}
\end{table}

\begin{figure}[t]
\begin{minipage}{0.45\linewidth}
\includegraphics[width=\linewidth]{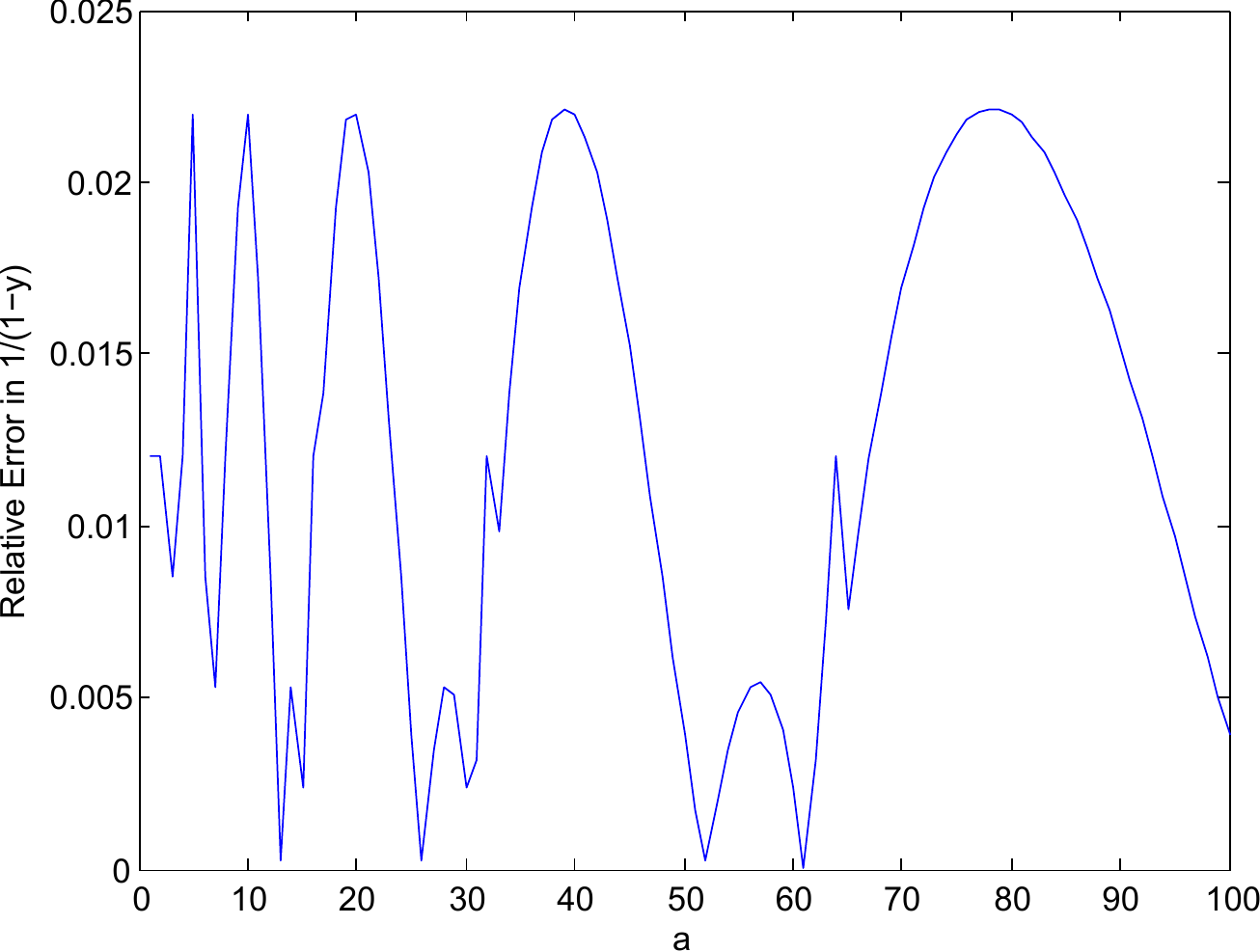}
\caption{Error in inversion using $R_2$ implemented using~\eq{notalpha} and~\eq{alpha} as a function of $a$ as discussed below.}
\end{minipage}
\hspace{1mm}
\begin{minipage}{0.45\linewidth}
\includegraphics[width=\linewidth]{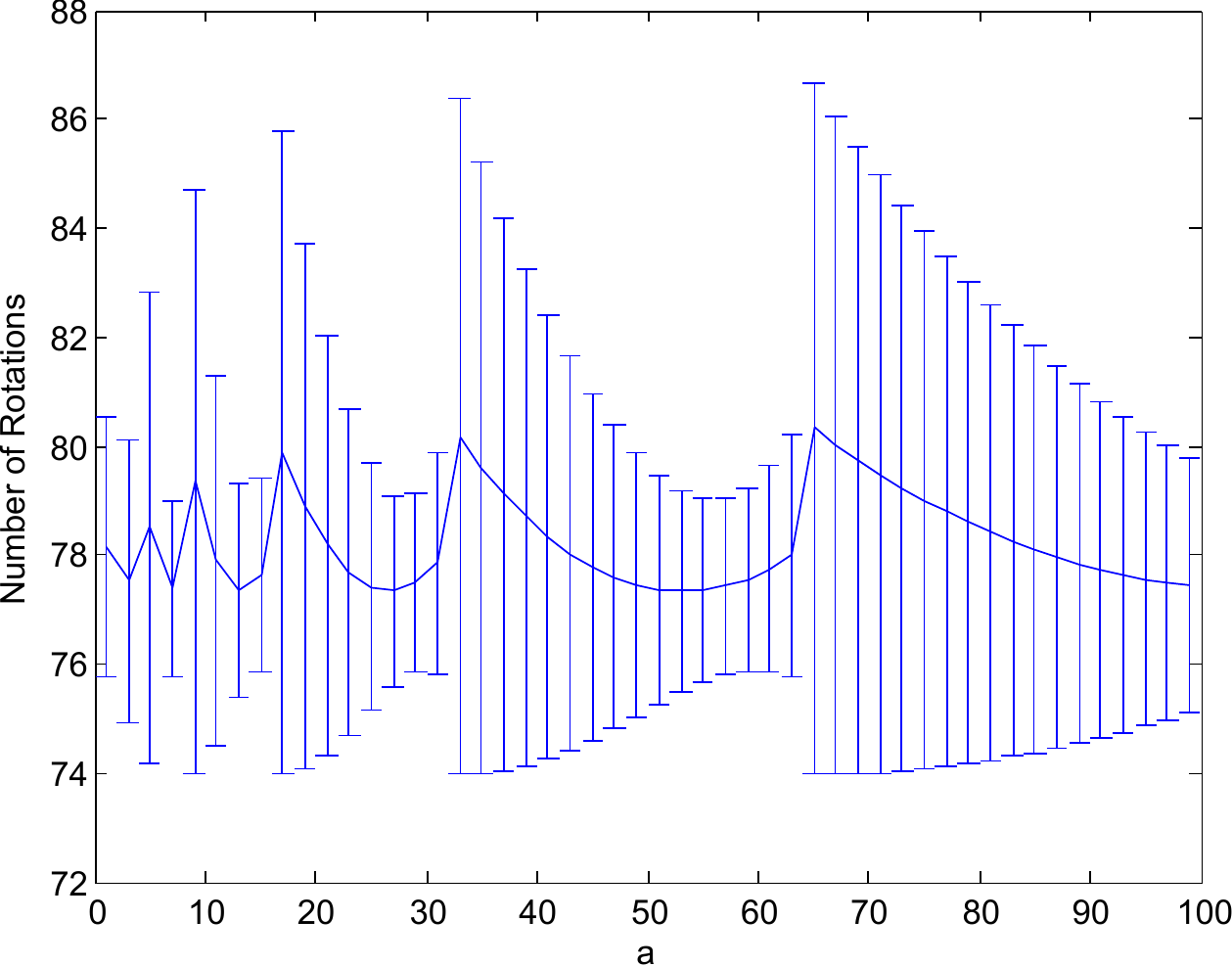}
\caption{Cost of implementing $R_2$ using~\eq{notalpha} and~\eq{alpha} as a function of $a$ as discussed below.  Error bars give the standard deviation.}
\end{minipage}
\end{figure}
We could directly implement these formulas for the (rescaled) reciprocal using the multiplication circuits discussed in~\sec{multiplication} for implementing $y^2$ but instead using gearbox circuits to implement an approximate squaring circuit.  We do this because $\GB$ is naturally an RUS circuit so it is much less costly to construct an approximate squaring circuit using these components then it is to use $\PAR$ circuits and oblivious amplitude amplification to convert them into RUS circuits.  

We also use another optimization that exploits the fact that $y\in[0,1/2]$.  This means that the worst--case errors can be minimized by using a Taylor--series expansion about $y=1/2$ rather than one centered about $y=0$.  In other words, we express $y^2$ as 
The second trick uses the fact that
\begin{equation}
y^2 = (y-1/4)^2 +\frac{y}{2} -\frac{1}{16},
\end{equation}
and use $\GB(y-1/4) \approx (y-1/4)^2$.

A further optimization that we consider is using high--order formulas for computing $(y-1/4)^2$.  The methods of~\sec{complete} can be used to show that
\begin{equation}
x^2 = \GB(x)- \GB(x,x,\arcsin(\sqrt{2/3})) -\GB(x,x,x,\arcsin(\sqrt{22/45})) + O(x^8).\label{eq:notalpha}
\end{equation}
If we cost all such inputs at one rotation, this approximant requires a minimum of $16$ rotations. Similarly, for any constant $\alpha\ge 2/\sqrt{15}\approx 0.52$.
\begin{equation}
\alpha x^2 = \GB(x,\arcsin(\sqrt{\alpha}))- \GB(x,x,\arcsin(\sqrt{\alpha^2 - \alpha/3})) -\GB(x,x,x,\arcsin(\sqrt{2\alpha^3/3-8\alpha/45})) + O(x^8).\label{eq:alpha}
\end{equation}
Thus $0.664355 (y-1/4)^2$ can be approximated to eighth order using  this approach using a minimum of $18$ rotations.  Since the requisite rotations are nearly half the size of those required for $y^2$, we do not need to expand about the midpoint to obtain sufficient accuracy for this rotation.  The approximant $R_2$ can then be realized by combining these ideas and noting that $\alpha y$ can be directly prepared from $\ket{a}$ for any constant $\alpha$ using standard synthesis methods.

Expansion into Chebyshev polynomials provides an excellent way to represent $1/(1-y)$ as a power series in $y$ that minimizes the max--norm of the difference between the approximant and the actual function.  Although such expansions can be practical and highly space efficient, we see that the coefficients in $R_2$, $R_4$ and $R_6$ do not remain small as the order of the polynomial approximation increases.  Although the Chebyshev approximation theorem clearly shows that such errors can be made arbitrarily small by increasing the order of the polynomial, the coefficients may increase with the approximation order.  This makes the complexity analysis much challenging since the cost of RUS arithmetic depends on these coefficients.  We will see below that such problems do not occur when the binomial method is used for division.

\subsection{The Binomial Method}
The binomial method is an alternative to Newton's method for computing the reciprocal that also has the property that it converges quadratically.  Here quadratic convergence means that the error drops doubly exponentially with the number of iterations used in the method.  At its heart, the binomial method is simply a re--grouping of the terms in the Taylor series of $(1-y)^{-1}$.  It reads
\begin{equation}
\frac{1}{1-y}  \approx (1+y)(1+y^2)\cdots (1+y^{{2^{n-1}}}),
\end{equation}
and the error in this approximation is at most $2^{-2^n}$.  The resultant series can be implemented as a series of multiplications of a form that is similar to that in~\cor{bigangle}.  

Quadratic convergence does not occur for direct implementations of the binomial method using RUS arithmetic.  This is because the product $(1+X)(1+Y)$ requires $O(1/(\epsilon^{1/(2k)}))$ copies of $X$ and $Y$ if we want to implement the multiplication within error $\epsilon$.  By iterating this process, it is then clear that $(1+y)(1+y^2)\cdots (1+y^{{2^{n-1}}})$ requires $e^{O(1/\epsilon^{1/(2k)})}$ copies of $y,y^2,\ldots,y^{2^{n-1}}$.  Thus linear, rather than quadratic convergence, is expected if the operations in the product are performed sequentially using RUS arithmetic.  Quadratic convergence can be recovered, however, in architectures that can prepare $e^{O(n)} = O({\rm polylog}(1/\epsilon))$ copies of the requisite states in parallel or by using caching strategies (which we discuss in~\sec{cache}).  In either case, the price of recovering quadratic convergence is a substantial increase in the width of the resulting circuits.

The following corollary shows that the binomial method can be used to compute the reciprocal in a remarkably space efficient manner.
\begin{corollary}
Given a qubit string $\ket{a}\in \mathbb{C}^{2^m}$ the rotation $\exp(-iX/a)$ can be approximated to within distance $\epsilon$ as measured by the $2$--norm using $\lceil\log_2( m+1 )\rceil+O(\log\log(1/\epsilon))$ additional qubits.
\end{corollary}
\begin{proof}
The error in the binomial approximation to the reciprocal is at most $2^{-2^n}$ for any $n$, hence if we wish the error to be $\epsilon$ then it suffices to take $n \in \Omega(\log \log(1/\epsilon))$.  This also suffices to make the error $O(\epsilon/n)$, which is sufficient to guarantee that the overall error is at most $O(\epsilon)$.  All of the monomials present in the binomial expansion can either be constructed using $\GB^{\circ p}(y)$ for $p=1,\ldots,n$ or directly implemented from the qubit representation of $y$.  This construction requires at most $n$ qubits.   Since this is an RUS circuit, we can make the error arbitrarily small using time--slicing without resorting to oblivious amplitude amplification, which substantially reduces the cost of method.
\cor{bigangle} can then be used to perform the multiplications using $M_2$, which requires only $2$ qubits within error at most $O(\epsilon/n)$.  Therefore the total error can be made at most $O(\epsilon)$ using the claimed number of qubits.

The remaining issue is that of constructing $y= 1-2^{-\lceil \log a \rceil}a$ from $\ket{a}$.  This problem is equivalent to bit shifting $a$ to the right until a number in the range $[1/2,1]$ is attained.  This, in effect, becomes the problem of preparing the state $\ket{a}\ket{\lceil \log_2 a \rceil}$.  This can be implemented using a simple reversible circuit, as illustrated below for the case of $4$ input qubits, where $\Qcircuit @R 1em @C 1.5em{&\gate{+j} &\qw}$ represents an adder circuit that increments a register by $j$.

\[
\Qcircuit @R 1em @C 1.5em {
&\qw & \ctrl{4}		&\qw		&\qw			&\qw		&\qw			&\qw		&\qw			&\qw		&\ctrl{1}	&\ctrlo{1}	&\ctrlo{1}	&\ctrlo{1}	&\qw			&\qw\\
&\qw &\qw			&\qw 		&\ctrl{3}		&\qw		&\qw			&\qw		&\qw			&\qw		&\ctrlo{1}	&\ctrl{1}	&\ctrlo{1}	&\ctrlo{1}	&\qw			&\qw\\
&\qw &\qw			&\qw		&\qw			&\qw		&\ctrl{2}		&\qw		&\qw			&\qw		&\ctrlo{1}	&\ctrlo{1}	&\ctrl{1}	&\ctrlo{1}	&\qw			&\qw\\
&\qw &\qw			&\qw		&\qw			&\qw		&\qw			&\qw		&\ctrl{1}		&\qw		&\ctrlo{4}	&\ctrlo{4}	&\ctrlo{4}	&\ctrl{4}	&\qw			&\qw\\
\lstick{\ket{0}}&\qw &\multigate{2}{+4}  &\ctrl{1}	&\multigate{2}{+3}	&\ctrlo{1}	&\multigate{2}{+2}&\ctrlo{1}	&\multigate{2}{+1}&\ctrlo{1}	&\qw		&\qw		&\qw		&\qw		&\multigate{2}{-1}  &\qw\\
\lstick{\ket{0}}&\qw &\ghost{+4}		&\ctrlo{1}	&\ghost{+3}		&\ctrl{1}	&\ghost{+2}		&\ctrl{1}	&\ghost{+1}		&\ctrlo{1}	&\qw		&\qw		&\qw		&\qw		&\ghost{-1}	&\qw\\
\lstick{\ket{0}}&\qw &\ghost{+4}		&\ctrlo{1}	&\ghost{+3}		&\ctrl{1}	&\ghost{+2}		&\ctrlo{1}	&\ghost{+1}		&\ctrl{1}	&\qw		&\qw		&\qw		&\qw		&\ghost{-1}	&\qw\\
\lstick{\ket{0}}&\qw &\ctrlo{-1}		&\targ		&\ctrlo{-1}		&\targ		&\ctrlo{-1}		&\targ		&\ctrlo{-1}		&\targ		&\targ		&\targ		&\targ		&\targ		&\ctrlo{-1}		&\qw\\
}
\]
This requires $\lceil \log_2 m+1 \rceil$ additional qubits (for technical reasons one additional qubit is also used in the reversible circuit that implements this, but this does not change the scaling).  Given this state, the input bit string $a$ can be logically bit shifted so that $\exp(-i2^{-\lceil \log a \rceil}aX )$ can be performed  from $\ket{a}$ using a series of controlled rotations.
Thus the rotation can be implemented within error $\epsilon$ using $\lceil\log_2( m+1 )\rceil+O(\log\log(1/\epsilon))$ qubits, as claimed.
\end{proof}

%In contrast to this result, conventional methods require $O(\log(1/\epsilon))$ qubits to express the result from the division algorithm, and hence there is an exponential separation between the number of additional qubits needed using our approach and the conventional method.  In particular, the implementation of Newton's method in \cite{CPP+13} requires roughly $3\lceil \log_2(1/\epsilon) \rceil$ qubits to implement a rotation, which for $1\%$ error corresponds to at least $21$ additional qubits.  This cost is rather substantial for a small quantum computer, so our space--efficient methods will likely prove to be more practical in near--future quantum computers.

%A natural criticism of this analysis is that we are only considering the space complexity of this task.  Indeed, these improvements in the space--complexity come at a price of exponentially worse scaling with $\epsilon$.  Nonetheless, the error tolerances required for the calculation of the reciprocal in HHL will typically be modest phase estimation within error $\delta$ requires $O(1/\delta)$ operations.  This exponential scaling can be mitigated, however, using phase estimation to cache precomputed results that can be subsequently cloned inexpensively.

\subsection{Comparison with classical methods for computing reciprocals}\label{sec:newton}

Similar to the case of multiplication, we provide a comparison between the resources required for RUS circuits for approximate computation of reciprocals with reversible circuits. We consider three different ways of implementing the computation of the reciprocal of $n$-bit integers: (i) a reversible implementation of the (extended) Euclidean algorithm following \cite{PZ:2003}, (ii) a reversible implementation of Newton's method following \cite{CPP+13}, and (iii) a table-lookup implementation. 

Similar to our comparison of conventional reversible circuits against RUS arithmetic for multiplication, we will choose a problem for which the inputs and outputs of both algorithms are comparable.  We take our problem to be one inspired by the linear systems algorithm~\cite{HHL09}.  A qubit string encoding a value $a$ is provided consisting of $n$ bits and we wish to use this to perform $e^{-i X/a}$.  Again we assume that all operations are decomposed into Clifford and $T$ gates and that the result is accurate to $n$ bits, meaning that the error in the resultant rotation is at most $2^{-(n+1)}$.

\subsubsection{Reciprocals via extended Euclidean algorithm} The basic idea is that for an input $x=\sum_{i=0}^{n-1} x_i 2^i$ the first $n$ digits of the reciprocal of $x$ can be computed by running the extended Euclidean algorithm for the computation of the greatest common divisor ${\rm GCD}(x, 2^n)$. By performing bit-shifts if necessary, we can assume that $x$ is odd, i.e., we are in the case where the GCD is equal to $1$. The extended Euclidean algorithm will then produce two integers $r$ and $s$ such that $rx + s 2^n = 1$. If $r = \sum_{i=0}^{n-1} r_i 2^i$, then the bit presentation of $x^{-1}$ is given by $\sum_{i=0}^{n-1} r_i 2^{-n+i}$. The computation of the extended Euclidean algorithm is highly non-trivial as the number of iterations of the basic reduction step depends on  the inputs. 

A reversible implementation of the extended Euclidean algorithm has been given that synchronizes the computation for any pairs of inputs so that it has the same number of steps is given in \cite{PZ:2003} and resources estimates are provided in \cite[Section 5.4.1]{PZ:2003}. The synchronization of the computation of the GCD of two $n$-bit number requires $4.5n$ repetitions of a fundamental cycle which in turn consists of the controlled application---depending on the content of a flag qubits---of $4$ adders, one swap, and one comparison circuit. Up to leading order, only the controlled adders matter and as a comparison can be reduced to an adder, we get an overall cost of $5$ times the cost for a controlled $n$-bit adder which we estimated earlier in  \sec{compMult} to be upper bounded by $84n$ in terms of $T$-gates. Hence we obtain $4.5n \cdot 5 \cdot 84n = 1,890 n^2$ as an upper bound for the number of $T$-gates. Putting everything together, we obtain 
the upper bound $2\cdot 1,890 n^2 + 1.15 (n+2) \log_2((n+2)/2^{(n+2)})$ on the number of $T$ gates. 
The space requirements for this method are bounded above by $5n+4\log_2(n)$ in \cite[Section 5.4.2]{PZ:2003} plus $1$ additional qubit for the output rotation.   

\subsubsection{Reciprocals via Newton's method} For a given $n$-bit integer $x$ one can obtain an approximation for $1/x$ by iterating the map $\mu: z \mapsto 2z - x z^2$ which quickly converges toward the fixed point of $\mu$ which is given by $\mu(1/x) = 1/x$. In order to achieve an output that has the first $n$ bits equal to those of of $1/x$, one needs at most $O(\log n)$ many iterations of $\mu$, when starting from an initial value that can be chosen to be any number between $2^{-n}$ and $2^{-n+1}$. A detailed analysis of the number of iterations to achieve a target precision of $n$ bits has been provided in \cite{CPP+13}: as the number of operations for the computation of $\mu$ in each iteration can be bounded by the circuit sizes for two multiplications (one for the square and one for multiplication with $x$; the multiplication by $2$ can be implemented as a bit shits), the number of iterations can be bounded by $2 \log_2(n)$ and there is an overhead of a factor of $2$ as the computation has to be reversed to disentangle the ancillae used in the algorithm. Overall, we obtain $4 \log_2(n)$ times the cost for one adder and two multipliers. We bound the $T$--count for an $n$-bit multiplier in \sec{compMult} above by $18n^2+18n$ and choose the adder to be the in-place adder from \cite[Table 1]{CDKM:2004} for which the cost in terms of Toffoli gates has been bounded by $2n-1$, i.e., we can bound the number of $T$-gates for this adder by $14n-7$. The total $T$--count is therefore the sum of all these contributions: $4 \log_2(n) \cdot (2(18n^2+18n) + (14n-7))= (144n^2+200n-28) \log_2(n)$ for the computation of $U_f$ as in \fig{comparison}

Putting everything together, we obtain 
$2\cdot (144n^2+200n-28) \log_2(n) +  1.15 (n+2) \log_2((n+2)/2^{(n+2)})$ many $T$-gates, where as before the leading factor of $2$ is due to the cleanup of the ancillas. For small values of $n$, the resulting upper bound estimates on the number of $T$-gates are shown in Table \tab{reciprocals}.

For the space requirements we obtain that in each iteration $n + 2(2n-1)$ new ancillas are required, so that in total we get an upper bound of $2 \log_2(n) (5n-2)=(10n-4) \log_2(n)$ for the computation of $U_f$. As in case of the multiplier we now have to compute the result into a rotation which leads to an additional overhead of $1$ qubit so that in total we need no more than $ \log_2(n)(10n-4)+1$ qubits. 

\subsubsection{Reciprocals via table-lookups} This leads to the same bound as in the case of the table-lookup implementation of the multipliers, with the only difference being that the input and output sizes are $n$. We get an upper bound of $2 \cdot (n 2^{n} (32n-84)) + 1.15 (n+2) \log_2((n+2)/2^{(n+2)})= 2^{n+1} (32n^2-84n) + 1.15 (n+2) \log_2((n+2)/2^{(n+2)})$ many $T$-gates that are required to implement the computation of the reciprocal of an $n$-bit numbers in constant space where $n\geq 5$. For $n=2$ we obtain the bound $8 \cdot 2 \cdot 7=112$ plus the cost for $Enc$ and for $n=4$ we obtain the bound $2^5 \cdot 5 \cdot 52=8320$ plus the cost for $Enc$. These results are summarized in \tab{reciprocals}.

As in case of the multiplier we obtain that the circuit can be implemented with an additional of $3$ qubits only. 

\subsubsection{Comparison with RUS methods}
A comparison between the resources required for reciprocals based on the Euclidean method, Newton's method, and the table-lookup method and $R_2$  is given in~\tab{reciprocals}.  The comparison was done using RUS synthesis to convert the rotations in $R_2$ into Clifford and $T$ circuits and using the same Toffoli gate used in previous steps.  $M_4$ and $M_6$ were used for $n=2$ and $n=4$ respectively to rescale the reciprocal after $1/(1-y)$ was computed using $R_2$, and in both cases $99\%$ of the error tolerance was used in the synthesis steps on the rotations in $R_2$ because these constitute the inner loop of the algorithm.  Making these rotations as inexpensive as possible helps reduce the cost of the overall algorithm because they are repeated many more times than the rescaling operation used after $1/(1-y)$ is computed.

  The first feature to note is that the number of ancilla qubits used in a reversible implementation of Euclid's algorithm and Newton's method are quite daunting.  Our estimates suggest that dozens to hundreds of qubits will be needed to perform a computationally useful quantum linear systems algorithm.  The Chebyshev approximation implemented using RUS synthesis methods, $R_2$, requires $6$ qubits in contrast to these results and requires a number of $T$--gates that is comparable to Euclid's method but substantially greater than Newton's method.  Unfortunately, because Chebyshev approximants have fixed accuracy, it is impossible to provide more than $4$ bits of precision.  Perhaps surprisingly, for the circuit sizes considered, table-lookup proved to be a viable approach because the table is only one dimensional.  We see that it is clearly a method of choice for small inversion problems, but its poor asymptotic scaling will make higher--order variants of $R_2$ much less expensive for more inversion problems with more stringent requirements on the error tolerance.  These results show that, unless further optimizations are used for division algorithms, the qubit requirements involved in performing an inversion are far beyond the capabilities of existing quantum computers and may exceed those of even near--future quantum computers.

\begin{table}[t!]
\begin{tabular}{c@{\qquad}c@{\quad}c@{\qquad}c@{\quad}c@{\qquad}c@{\quad}c@{\qquad}c@{\quad}c}
\hline\\
Reciprocal & \multicolumn{2}{c}{$n=2$\phantom{111}} & \multicolumn{2}{c}{$n=4$\phantom{111}} & \multicolumn{2}{c}{$n=8$\phantom{111}} & \multicolumn{2}{c}{$n=16$} \\[0.5ex]
method     & $T$--count & qubits & $T$--count & qubits & $T$--count & qubits & $T$--count & qubits \\[1.5ex]
\hline\\
Euclid & 1.51E+04 & 12 &  6.05E+04 & 23 & 2.42E+05 & 44 & 9.68E+05 & 85 \\[1.5ex]
Newton & 1.92E+03 & 17 & 6.21E+03 & 73 & 2.17E+04 & 229 & 8.05E+04& 625\\[1.5ex]
Table-lookup & 1.40E+02 & 3 & 8.38E+03 & 3 & 7.05E+05 & 3 & 8.98E+08 & 3 \\[1.5ex]
$R_2$ &3.17E+03& 6&1.53E+05 &6 &NA & &NA & \\[1.5ex] 
%$R_4$ & & & & & & & & \\[1.5ex] 
%$R_6$ & & & & & & & & \\[1ex] 
\hline
\end{tabular}
\caption{\label{tab:reciprocals} A comparison of the resources required for space efficient computation of reciprocals on a quantum computer. Shown are mean circuit sizes (number of $T$-gates) and number of required qubits for $n=2, 4, 8, 16$ bit numbers.  We assume $a\in (1,2]$ for simplicity in the case of $R_2$.}
\end{table}

\section{Caching Strategies}\label{sec:cache}
One of the drawbacks of the approaches we outline above is that arithmetic elements designed using the previous methods do not necessarily compose nicely.  For example, if function $f_1$ requires $N_1$ rotations to implement and $f_2$ requires $N_2$ rotations then $f_1 \circ f_2$ requires $N_1N_2$ rotations.  It may be in principle natural to unravel the recursion and approximate the resulting function using RUS arithmetic, but in practice methods such as Newton's method or the Binomial method that explicitly use recursion and avoiding an exponential slowdown at the price of increased circuit width may be desirable in such cases.  Such exponential slowdowns can be avoided by using phase estimation to ``cache'' the result of the $f_2(x)$ in a register before using that result in $f_1(x)$.  These results are also valuable for cases where is is useful to output the result of the RUS arithmetic as a qubit string.

\begin{lemma}\label{lem:cache}
Let $f_1$ be a differentiable function that requires $N_1$ rotations to implement and satisfies $|\partial_\phi f_1(\phi)|\le \kappa$ and let $f_2$ be a piecewise continuous function that requires $N_2$ rotations to implement.  Then $e^{-i(f_1( f_2(\phi))X)}$ can be implemented within error $\epsilon$ with probability at least $1-\delta$ using $$\lceil\log_2(\kappa/\epsilon) \rceil + \lceil \log_2(2+1/(2\delta)) \rceil$$ 
additional qubits and 
$$O\left(\frac{N_2 \kappa }{\epsilon \delta} + N_1 \right)$$
additional single qubit rotations.
\end{lemma}
\begin{proof}
In order to ensure that the total error is at most $\epsilon$ for the computation of $f_1(f_2(\phi))$ it suffices to guarantee that $f_2(\phi)$ is computed within error $\epsilon$ because Taylor's remainder theorem implies that
\begin{equation}
|f_1(f_2(\phi)) - f_1(f_2(\phi) + \epsilon/\kappa)| \le \frac{\epsilon}{\kappa}\max_x |f'(x)| \le \epsilon,
\end{equation}
by assumption.  

Instead of viewing $f_2$ as faulty, imagine instead that $f_2$ is given by a qubit string that stores the result of approximating its value using phase estimation.  Given a unitary that enacts $e^{-i f_2(\phi)X}$ using $N_1$ fundamental rotations, it then follows that the following unitary can be implemented using $N_1$ rotations
\begin{equation}
H e^{-i f_2(\phi )X} H\ket{0} = e^{-i f_2(\phi)}\ket{0}.
\end{equation}
Thus phase estimation can be used to estimate $f_2(\phi)$ within error $\epsilon/\kappa$ using $t$ qubits and success probability $1-\delta$ where $t$ is~\cite{NC00} 
\begin{equation}
t=\lceil\log_2(\kappa/\epsilon) \rceil + \lceil \log_2(2+1/(2\delta)) \rceil.\label{eq:71}
\end{equation}
Phase estimation requires $2^t$ applications of $e^{-if_2(\phi)X}$ in order to learn the state within the required tolerances.  This means that the cost of preparing $\ket{f_2(\phi)}$ is given by~\eq{71} to be $O(N_2 \kappa/(\epsilon \delta))$.  Now if we cost performing $e^{-i f_2(\phi)X}$ using this state as a single rotation then the cost of performing $e^{-i f_1(f_2(\phi)) X}$ is $N_1$.  The lemma then follows by summing both terms.
\end{proof}

\lem{cache} shows that caching can reduce the cost of using RUS arithmetic to approximate recursive functions from exponential to linear at the price of additional costs incurred by using phase estimation.  Such tradeoffs are asymptotically beneficial when $\frac{\epsilon\delta}{\kappa}\in o(N_1)$.  Since recursion can be quite expensive in terms of either time or space it may often be of practical interest to find alternatives to these approaches, such as that given below.

\section{Square Wave Synthesis}\label{sec:square}
An alternative approach to the method described above is to synthesize functions using the square wave property of the $\GB$ function.   The idea behind our approach builds upon existing classical results that use \emph{exact} square waves to approximate a function~\cite{SMM08}.  Assume a function $f(x)$ is provided and that the objective is to approximate it on a uniform mesh of points on $[x_{\rm min},x_{\rm max}]$ and define  $\{x_i:i=1,\ldots, N\}$ to be the midpoints of each segment of the mesh.  An approximation to the function of the form
\begin{equation}
f(x) \approx \sum_{j=1}^N a_j S_j(x),\label{eq:fapprox1}
\end{equation}
is then sought where each $S_j(x)$ is a square wave with period $T_j$.  We achieve this by demanding that the approximant equals $f(x)$ at each midpoint $x_i$
\begin{equation}
f(x_i) = \sum_{j=1}^N a_j S_j(x_i).
\end{equation}
The values of $a_j$ are then found by solving the resultant system of equations.  The error in this piecewise approximation is therefore at most 
\begin{equation}
\left|f(x) - \sum_{j=1}^N a_j S_j(x)\right|\le \max_x |f'(x)|(x_{\rm max}-x_{\rm min})/(2N).\label{eq:squareerror}
\end{equation}

Now let us consider a quantization of this problem.  Assume that we have a qubit string, $\ket{x}$, that encodes the number $x$ and we want to implement
\begin{equation}
\ket{x}\ket{0} \mapsto \ket{x}e^{-if(x) X}\ket{0}.
\end{equation}
As an intermediate step, note that if $\ket{x}$ is given then $\ket{x}e^{-i ax X}\ket{0}$ can be implemented for any fixed $a\in \mathbb{R}$ using traditional circuit synthesis methods and rescaling the rotations used in implementing the rotations that are conditioned on the value of $x$.  For example,   if $x$ is represented using three qubits then $\ket{x}e^{-i ax X}\ket{0}$ can be implemented as
$$
\Qcircuit @R 1em @C 1.5em {
\lstick{\ket{x_1}}	&\ctrl{3}		&\qw			&\qw			&\qw&\\
\lstick{\ket{x_2}}	&\qw			&\ctrl{2}		&\qw			&\qw&\\
\lstick{\ket{x_3}}	&\qw			&\qw			&\ctrl{1}		&\qw&\\
			&\gate{e^{-ia X}}	&\gate{e^{-i2aX}}	&\gate{e^{-i4aX}}	&\qw&\\
}
$$
This observation will be useful because it shows that we can rescale inputs by a constant factor without using RUS arithmetic.

There are a few slight complications that arise when using $\GB$ to approximate the square waves that appear in the synthesis method.  Firstly, instead of square waves we use
\begin{equation}
\frac{4}{\pi}(\GB^{\circ k}(x\pi +\pi/4)-\pi/4) \approx H_{\rm Heaviside}(x),
\end{equation}
and replace each term in~\eq{fapprox1} with
\begin{equation}
\GB\left(\arctan\left(\sqrt{\tan\left({2a_j}\right)}\right),\GB^{\circ k}(x\pi/T_j +\pi/4)\right)-a_j,\label{eq:nearlin}
\end{equation}
for $a_j \le \pi/8$.  This can be verified by observing that when $\GB^{\circ k}(x\pi +\pi/4)$ evaluates to$\pi/2$ the function evaluates to $a_j$ as required.  Larger values of $a_j$ can also be straight forwardly implemented by summing several gearbox circuits that use smaller $a_j$.
\begin{figure}[t!]
\includegraphics[width=0.7\linewidth]{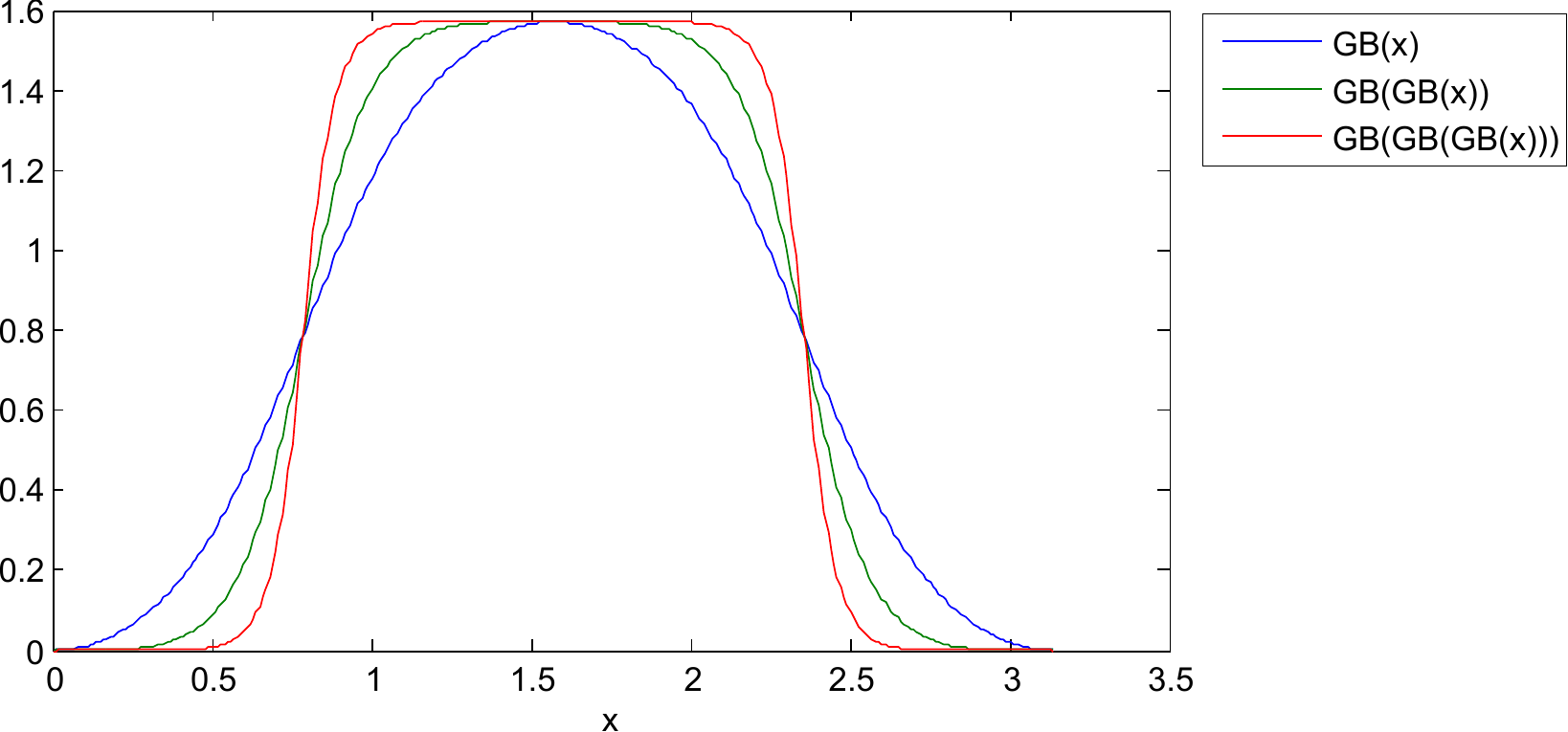}
\caption{$\GB^{\circ k}(x)$ for $k=1,2,$ and $3$. The functions approach a square wave as $k\rightarrow\infty$.}
\end{figure}

The approximation procedure then works as follows.  We first linearize the expressions and find $a_j$ such that for each midpoint $x_i$,
\begin{equation}
f(x_i) = \sum_{j=1}^N a_j \left(\frac{4}{\pi}(\GB^{\circ k}(x_i\pi/T_j +\pi/4)-\pi/4)\right).\label{eq:linear}
\end{equation}
The product of the $a_j$ with each of the square waves could be approximated using the multiplication circuits given previously.
A much better approach is to approximate the near--linear function of $a_j$ given in~\eq{nearlin} as
\begin{equation}
f(x) \approx \sum_{j=1}^N \GB\left(\arctan\left(\sqrt{\tan\left({2a_j}\right)}\right),\GB^{\circ k}(x\pi/T_j +\pi/4)\right)-a_j,
\end{equation}
where the periods for each square wave obey $T_j= 2(x_{\rm max} - x_{\rm min}) j/N$.
The resultant circuit does not require time slicing in order to increase the approximation accuracy, unlike the previous methods.  A formal description of the algorithm is given in~\alg{square}.
It is worth noting that square waves can also be generated using the analogous property of $\PAR$, but we favor $\GB$ in this construction because it is naturally an RUS circuit.

There are some issues that arise from this approximation.  The first is results of composed gearbox circuits do not approximate the square wave well near each discontinuity.  In particular, $\GB^{\circ k}(x+\pi/4) \approx \pi/4 +2^k x$ whereas $\GB^{\circ k}(x) \approx x^{2^k}$.  Therefore we expect the square wave to jump to its maximum value over a range of $2^{-k}$ radians whereas the gearbox circuit will be in near perfect agreement with the square wave near $x=0$.  This means that $k\ge 8$ may be needed to appreciably reduce the errors due to non--instantaneous ramp up near $\pi/4$.  Since the expected number of rotations for $\GB^{\circ k}$ is roughly $2^{k-1/2}$~\cite{WK13} this implies that these errors are can be reduced at linear cost.  

Another issue facing square wave synthesis is that each of the approximate square waves achieves a value of $0$ near $x=0$. If the function is continuous then there exists an $\epsilon$ neighborhood about $0$ such that for any finite $k$ the approximation error is $O(1)$, unless the function to be approximated also obeys $f(0)=0$.  Similar problems can occur at the end of the approximation interval.  Both such issues can be addressed by padding the approximation interval from $[x_{\rm min},x_{\rm max}]$ to $[x_{\rm min}-\Delta,x_{\rm max}+\Delta]$ for some $\Delta >0$.

\begin{algorithm*}[t!]
\caption{\label{alg:square} Square Wave Synthesis Algorithm.}
\rule{\linewidth}{1pt}
\begin{algorithmic}
\Require Function $f$, endpoints for interpolation range $x_{\rm min}$ and $x_{\rm max}$, number of square waves used in the synthesis $N$, $k$ recursion order for Gearbox circuit.
\Ensure  $\{a_j:j=1,\ldots,N\}$ such that $f(x) \approx \sum_j a_j \frac{4}{\pi}(\GB^{\circ k}(xN\pi/(2j(x_{\rm max}-x_{\rm min})) +\pi/4)-\pi/4)$.

\vskip0.2em
\hrule
\vskip0.2em

\Function{squareWaveApprox}{$f$, $x_{\rm min}$, $x_{\rm max}$, $N$}
\State $x_j \gets x_{\rm min} + (x_{\rm max}-x_{\rm min})(j-1/2)/N$.
\State $T_j \gets 2(x_{\rm max}-x_{\rm min})j/N$.
\State $A_{j,k} = \frac{4}{\pi}(\GB^{\circ k}(x_j\pi/T_k +\pi/4)-\pi/4)$.
\State $\vec y \gets f(x_j)$.\inlinecomment{Compute function on mesh}
  \State \Return $A^{-1} \vec y$.\inlinecomment{Solve system of equations, $A\vec a = \vec y$, for coefficients $a_j$}
  %\inlinecomment{We must normalize the updated weights before returning.}
\EndFunction
\end{algorithmic}
\rule{\linewidth}{1pt}
\end{algorithm*}

We will now illustrate these ideas by using approximate square waves to implement $1/(1-y)$ on $y\in [-0.1,0.6]$.  We use $k=8$ and take $N=71$ and find from the data in~\fig{sqwavesynth} that the function can be approximated on the subinterval $y\in [0,0.5]$ to within a maximum error of $2.1\%$ and a mean error of $0.36\%$.  These spikes in the error occur because  if $y$ is in the interval where the function is ramping up from $-1$ to $1$ then the midpoint approximation will fail to be very accurate.  The probability of such events can be decreased by increasing $k$ at the price of increasing the circuit size.  

The cost of this synthesis process is $O(2^k N\log N)$.  However, a major advantage of this method is that the rotation--depth of the resulting circuit is substantially smaller: $O(kN\log N)$~\cite{WK13}, meaning that the overwhelming majority of the rotations used can be performed simultaneously in an architecture that allows parallel execution.

As a final note, the square wave method should not just be seen as a method for implementing a piecewise constant function.  A broader view is to see each square wave as a logical function that thresholds the rotation angles input into the gearbox circuit.  These circuits therefore give a way to approximate logical operations on rotation angles and hence square wave synthesis can be viewed as a special case of logical operations being used to implement a function via a lookup table on $N$ points with piecewise constant interpolation between them.  

\section{Entangled Inputs}\label{sec:entangle}
We focused in the previous examples on cases where the input is stored as a single qubit state.  More generally, the qubits fed into the gearbox circuit could be entangled with another register.    For example, rather than taking the qubit state $\ket{\phi}$ as input to a gearbox could consider input $\sum_a \alpha_a\ket{a}\ket{\phi_a}$ for some sets of states $\{\phi_a\}$ and amplitudes $\{\alpha_a\}$.  Such cases are not pathological: they are representative of typical use cases in linear systems algorithms or quantum simulation algorithms.  We will see that the analysis of such cases is much more complicated than the case of unentangled inputs considered in~\sec{complete}.

It is tempting to imagine that our methods will run without modification given entangled inputs.  After all, if a failure occurs then the correction operation is a fixed Clifford gate regardless of the value of $\phi_a$.  Thus regardless of whether the input is entangled or not, the operation $e^{-if(\phi_a) X}$ is applied to the target qubit $\ket{\psi}$.  The problem with this reasoning is that the probability of success or failure is intimately linked with the value of $\phi_a$.  This means that if success is measured then amplitudes of all states that lead to low success probability will be suppressed; whereas if failure is measured then terms with large $\alpha_a$ will drop in amplitude.  This is because of the information disturbance property of quantum mechanics.  Thus if we want to apply our RUS circuits to such states without disturbing the distribution we must guarantee that success or failure is insufficiently informative for the coefficients to be significantly altered by the measurements used in RUS circuits.

As an example, let us focus on the case of applying a gearbox circuit given a set of entangled inputs.  The gearbox circuit has success probability $\cos(\theta)^4 +\sin(\theta)^4$ given input $e^{-i\theta X}$.  This means that if entangled inputs are provided then a state proportional to
\begin{equation}
\sum_a \alpha_a \GB(\phi_a) \ket{\psi} = \sum_a \alpha_a \sqrt{\cos(\phi_a)^4 + \sin(\phi_a)^4} \alpha_a e^{-i\tan^{-1} (\tan^2(\theta_a))X}\ket{\psi},
\end{equation}
is prepared upon success.  

In cases where almost all of these probabilities are equal, the multiplicative factor due to the success probability drops out in normalization.  Although this may seem like a trivial case, there are important examples where this actually happens.  For example, the circuit $\PAR(\phi_a)$, which is the original PAR circuits of~\cite{JWM+12}, satisfies this and hence works for entangled inputs without modification.  Thus balancing all probabilities to be nearly equivalent will ensure that the RUS circuits will not have a meaningful impact on the result.  In other cases the probability distribution may be naturally be flat enough such that these errors can be neglected.

The problem of converting $\GB$ and $\PAR$ to circuits that leak negligible information about the $\phi_a$ can be solved by using product formulas, oblivious amplitude amplification and time--slicing.  We will focus on $\GB$ first.  For small input angles, the success probability scales as $\cos^4(x)+\sin^4(x)=1-2x^2+(8/3)x^4+O(x^6)$. These probabilities approach $1$ as $x$ approaches zero.  Unfortunately, we cannot implement a squaring circuit using time--slicing from this alone because the angles implemented also scales as $O(x^2)$.  In order to make this approach work we will have to construct a new function that is flatter.  Consider the following function

\begin{equation}
\GB\left(\frac{x}{2}+\pi/4\right) \GB\left( i(\frac{x}{2} +\pi/4)\right)\GB(x)\equiv\GB(x),
\end{equation}
where $\GB(ix)$ denotes the gearbox circuit with a controlled $iX$ gate rather than a controlled $-iX$ gate.  This is equivalent to the gearbox circuit because the two left most operations invoke equal and opposite rotations upon success.  The probability of all three of these operations succeeding on the first attempt is
\begin{equation}
(\cos(x/2+\pi/4)^4+\sin(x/2+\pi/4)^4)^2(\cos(x)^4+\sin(x)^4)=\frac{1}{4}-\frac{1}{4}x^4 +O(x^6).
\end{equation}
For $x\le 0.1$, the deviation from success probablity $1/4$ is roughly $10^{-5}$.  Since the deviation from uniformity scales like $O(x^4)$ here, we can take $x\rightarrow x/r$ and use $r^2$ time slices to reduce the scaling of the skewness of the distribution to  $O(x^4/r^2)$ while retaining the initial rotation angle (and at the same time improving its fidelity with an ideal squaring operation) if all such measurements succeed.

If we forgo measurement and mark all states that would correspond to three successful measurements as a ``good'' state then it is straight forward to see from~\cite{BCC+13} that a single iteration of oblivious amplitude amplification will make the probability of success $1-O(x^4/r^2)$.  Thus by slicing we can make this probability of success arbitrarily high while minimizing the information leaked.

A similar expansion also exists for $\PAR$:
$$\GB(a/2+\pi/4)\GB(b/2+\pi/4)\PAR(a,b),$$
which leads to success probability $1/4 + (a^2b^2-b^4-a^4)/4 +O(x^4)$ where $x\ge \max\{|a|,|b|\}$. The two gearbox circuits used here do not produce rotations that cancel each other out.  Instead, by choosing an even number of timesteps, these rotations can be canceled in subsequent steps; thereby removing any residual information about the values of $a$ and $b$ while retaining $25\%$ success probability.

Thus $\GB$ and $\PAR$ can in principle be applied in cases where entangled inputs are used.  This raises the possibility that these constructions might be able to substantially reduce the space requirements for linear systems algorithms.  Formally analyzing the time complexity of this application would, however, require carefully trading off both errors in the amplitudes as well as the errors in implementing $\ket{a}e^{-iX/a}\ket{0}$ for each $\ket{a}$.  We leave analyzing such applications for subsequent work.

\section{Conclusion}
Our main contribution is a fundamentally quantum approach to arithmetic that stores both inputs and outputs not as qubit strings but as amplitudes.  More specifically, our method represents the value $\phi$ as $e^{-i\phi X}\ket{0}$ and uses measurement as an active participant in the approximations.  This strategy is profitable because of the use of repeat until success circuits, which removes the possibility of a faulty measurement irreparably corrupting an entire multi--step computation.  In a deep sense, our work can be thought of as providing a generalization of existing circuit synthesis ideas wherein axial rotations of the form $e^{-i f(\phi_1,\ldots, \phi_k) X}$ are implemented instead of a constant function of the form $e^{-i f X}$.  

Optimizing the performance of our circuits therefore requires a different flavor of numerical analysis that takes into account the ability of measurement to non--deterministically apply non--linear transformations on the input data.  We provide a first foray into this field of quantum numerical analysis by providing Taylor series based methods for finding arbitrarily high--order approximations to smooth functions using repeat until success circuits.  Specific examples are given for multiplication and division wherein we find that we can implement both operations using substantially fewer operations than the most popular methods currently proposed.  In particular, our methods can provide an arbitrarily accurate multiplier using a constant number of ancilla qubits.  Reciprocals can also be implemented in this manner using a number of qubits that scales logarithmically in the number of bits of precision required.  This is significant because reversible circuits for division or multiplication can require hundreds of logical qubits, making algorithms like linear systems algorithms outside of the reach of small quantum computers.  Our ideas promise to enable such applications of quantum computers.

We also introduce another method for synthesizing functions using square waves.  This approach does not take advantage of the smoothness of the function but instead uses the fact that our repeat until success circuits approach a square wave upon being recursed many times.  This naturally leads to a piecewise approximation to a function.  We show that this method also yields a very space efficient approximation to a circuit and further leads to approximations that behave well even in the presence of discontinuities.

An important consideration that we have not discussed here is the issues that emerge due to implementing RUS arithmetic on a faulty quantum computer as opposed to an ideal one.  In fault tolerant constructions, more physical qubits will be needed to encode increasingly accurate approximations to ideal RUS arithmetic.  This means that RUS arithmetic cannot be performed within arbitarily small error using a fixed number of physical qubits in such cases.  On the other hand, reversible circuits can require an increasing number of \emph{logical} qubits to perform elementary arithmetic operations.  Each logical qubit must also be formed using many physical qubits and hence the number of physical qubits required for accurate computation will increase rapidly with the desired error tolerance.  Although this tradeoff may typically favor RUS approaches, more work will be needed to fully assess the costs of RUS arithmetic and function synthesis within the context of fault tolerant quantum computation.

Looking forward, there are several potential further avenues of inquiry revealed by this work.  There may exist other classes of repeat until success circuits that are yet to be found that could enable new ways of synthesizing functions.  This may lead to more efficient, or more natural, ways of implementing arithmetic in our paradigm.  In a similar vein, synthesizing functions by way of a Taylor series expansion is perhaps not the most natural way to implement these rotations since multiplication is not natural for large angles in our framework.  The gearbox and generalized PAR circuits  proposed here may also have applications beyond function synthesis.  Also the question of whether sub--polynomial scaling with $\epsilon^{-1}$ can be obtained with this method also remains open. Finally, a natural extension of our work on square wave synthesis is to consider a decomposition of $f$ into a sum of approximate Walsh functions, which would constitute a natural way to approximate arbitrary functions on a quantum computer without using a polynomial expansion.  This list is by no means exhaustive and we suspect that the ideas presented here may stimulate the development of new and innovative methods for performing arithmetic, function approximation or state distillation on fault--tolerant quantum computers.
\begin{figure}[t!]
\includegraphics[width=0.6\linewidth]{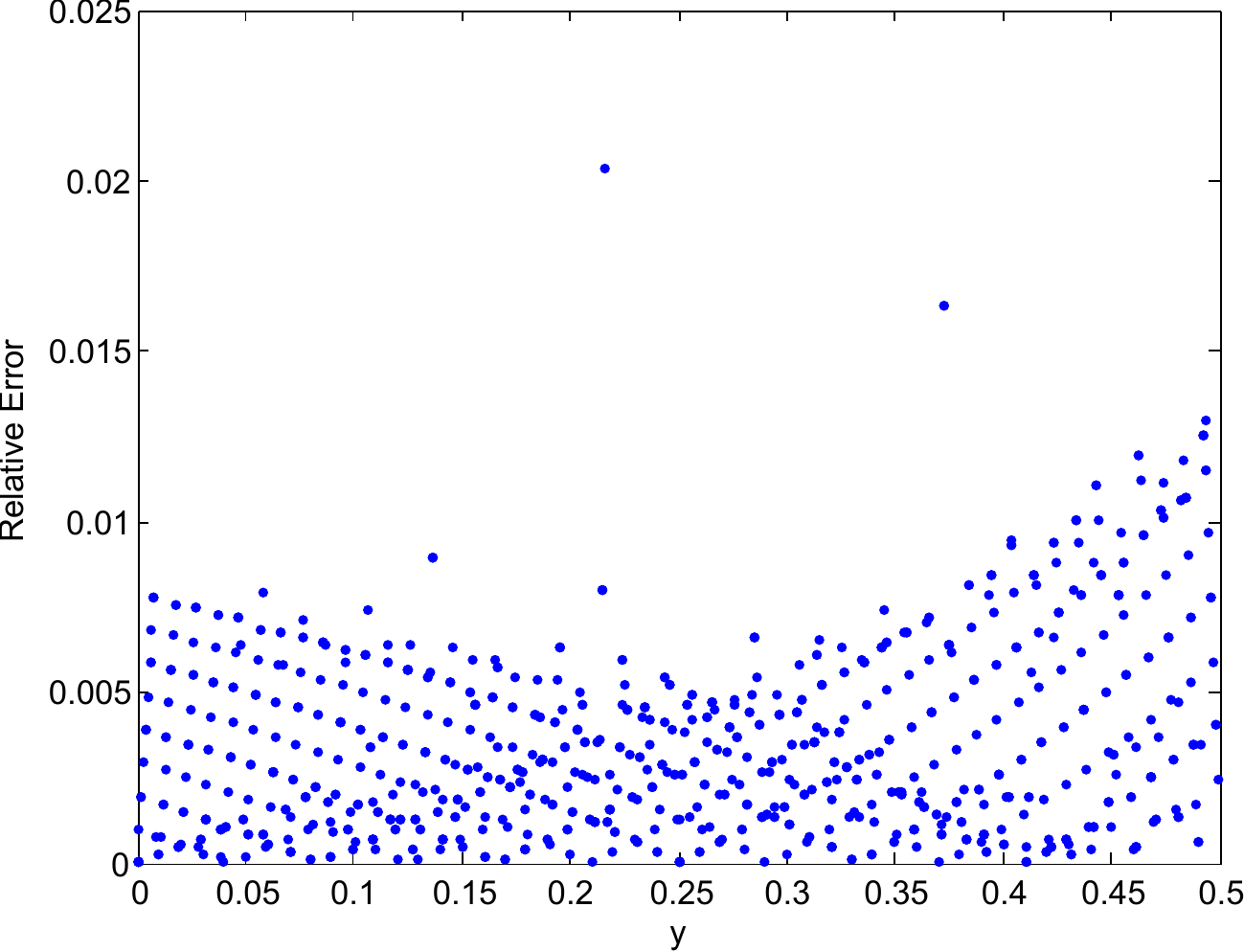}
\caption{Plot of the relative error in square wave synthesis of $1/(1-y)$ for $k=8$ and $N=71$ for values of $y$ corresponding to $a\in [1,1024]$.\label{fig:sqwavesynth}}
\end{figure}

\appendix
\section{Proof of~\lem{PAR}}\label{app:Lemma2}
The proof of~\lem{PAR} requires the use of oblivious amplitude amplification~\cite{BCC+13}, which is a technique that allows amplitude amplification to be performed without knowing the initial state that the system is prepared in.  We provide the result, proven in~\cite{BCC+13}, in the following lemma.
\begin{lemma}(Berry, Childs, Cleve et al.)\label{lem:BCC+}
Let $U$ and $V$ be unitary matrices on $n+1$ qubits and $n$ qubits, respectively, and let $\theta\in (0,\pi/2)$.  Suppose that for any $n$-qubit state $\ket{\psi}$ 

$$U\ket{0}\ket{\psi} = \sin(\theta) \ket{0} V\ket{\phi} +\cos(\theta)\ket{1}\ket{\phi}, $$
where $\ket{\phi}$ is an $n$-qubit state.  Let $R:= (\ketbra{0}{0} - \ketbra{1}{1})\otimes \openone$ and $S:=-URU^\dagger R$.  Then for any $t\in \mathbb{Z}$,
$$
S^t U \ket{0} \ket{\psi} = \sin((2t+1)\theta)\ket{0}V\ket{\psi} + \cos((2t+1)\theta)\ket{1}\ket{\phi}.
$$
\end{lemma}
In our context the $R$ operation in the above lemma can simply be taken to be the $Z$ gate.  This result naturally leads to the following corollary:

\begin{corollary}\label{cor:RUS}
Let $U: \ket{0}\ket{\psi} \mapsto \frac{\ket{0}V\ket{\psi}+ \ket{1}\ket{\phi}}{\sqrt 2}$ then the state $V\ket{\psi}$ can be performed deterministically using three applications of $U$ and a constant sized Clifford circuit.
\end{corollary}
\begin{proof}
Oblivious amplitude estimation cannot be used directly to see this result because $\theta=\pi/4$ in this case.  Since the amplitudes of the ``good'' state vary like $\sin((2t+1)\theta)$, it is clear that amplitude is not transferred to the good solution by repeatedly using oblivious amplitude amplification.  Instead we poison the success probability by defining a new success condition that would happen with probability $1/4$ if amplitude amplification was not used.

\begin{equation}
\tilde{U}\ket{0}^2\ket{\psi}:= H\ket{0} \otimes (U\ket{0}\ket{\psi})
\end{equation}
If the top qubit in the above circuit were to be measured then it would yield the desired transformation on the bottom qubit with probability $1/4$.  This means that in essence, these transformations allow us to implement a transformation that logically has the same action upon success but with half the original success probability.  

This is all that we need in order to make the application of $V$ to $\ket{\psi}$ deterministic.
For $\ket{\tilde{\psi}}:= \ket{0}^2\ket{\psi}$ we see that we can express 
\begin{equation}
\tilde U\ket{0}\ket{\tilde\psi} = \frac{1}{2} \ket{0}\left(\ket{0}{V}\ket{{\psi}}\right)+ \frac{\sqrt{3}}{{2}}\ket{1}\ket{\tilde{\phi}}.
\end{equation}
for a $n+1$ qubit state $\ket{\tilde{\phi}}$.  We then define a successful outcome to be any measurement where the first two qubits are $0$, similar to the discussion in~\cite{BHM+00}.  Specifically, after applying possible simplifications, the amplitude amplification circuit takes the form:
\[
\tilde UR\tilde U^\dagger R\tilde U \ket{0}^2 \ket{\psi}:= \qquad\qquad\Qcircuit @R 1em @C 1.5em {
\lstick{\ket{0}}		&\gate{H}			&\qw			&\gate{-Z}		&\gate{H}&\gate{-Z}&\gate{H}&\qw\\
\lstick{\ket{0}}		&\multigate{1}{U} 	&\qw			&\ctrlo{-1}		&\multigate{1}{U^\dagger}&\ctrlo{-1}&\multigate{1}{U}&\qw\\
\lstick{\ket{\psi}}	&\ghost{U}		&\qw 		&\qw			&\ghost{U^\dagger}&\qw&\ghost{U}&\qw
}
\]
This is exactly in the form required by~\lem{BCC+} with $\theta= \arcsin(1/2)=\pi/6$.  Thus we see from \cite[Lemma 3.6]{BCC+13} that
\begin{align}
-\tilde UR\tilde U^\dagger R\tilde U \ket{0}^2 \ket{\psi}&= \sin((2+1)\pi/6)\ket{0}\left(\ket{0}{V}\ket{{\psi}}\right)+\cos((2+1)\pi/6)\frac{\sqrt{3}}{{2}}\ket{1}\ket{\tilde{\phi}}\nonumber\\
&=\ket{0}^2{V}\ket{{\psi}}.
\end{align}
Thus the success probability can be boosted to $100\%$ using three calls to $\tilde{U}$ and a constant number of Clifford operations, which is equivalent to three calls to $U$ and a constant number of Clifford gates as claimed.
\end{proof}
This corollary is significant because it allows us to make any such circuit, such as those used in PAR circuits or non--deterministic circuit synthesis methods \cite{PS13,BRS14}. 
In particular, \cor{RUS} leads directly to a proof of~\lem{PAR}.
\begin{proofof}{\lem{PAR}}
Consider $\PAR(\theta_1,\ldots,\theta_k)$.  Eq.~\eq{GHZeq} shows that the generalized~$\PAR$ circuits prepare the state

\begin{align}
\PAR(\theta_1,\ldots, \theta_k): &\ket{0}\ket{0}^{k-1}\ket{\psi} \nonumber\\
&\mapsto {\rm GHZ}^{-1}\left(e^{-i\theta_1 X} \otimes \cdots \otimes e^{-i\theta_k X}\ket{0}^{k} \ket{\psi} - \prod_{j=1}^k \cos(\theta_j)\ket{0^k}\ket{\psi} -(-i)^k \prod_{j=1}^k \sin(\theta_j) \ket{1^{k}}\ket{\psi}\right) \nonumber\\
&\qquad+ \frac{1}{\sqrt 2}\ket{0^k}(\cos(\theta_1)\cdots \cos(\theta_k) \openone -i \sin(\theta_1)\cdots \sin(\theta_k)X)\ket{\psi}\nonumber\\
&\qquad+ \frac{1}{\sqrt 2}\ket{1}\ket{0^{k-1}}(\cos(\theta_1)\cdots \cos(\theta_k) \openone +i \sin(\theta_1)\cdots \sin(\theta_k)X)\ket{\psi}.\label{eq:betaeq}
\end{align}
For notational simplicity in the following we will represent a number $v$ as a bitstring $v_1\cdots v_k$, where products within the ket of the form $\ket{v_1\cdots v_k}$ are interpreted to be the concatenation of the bits so that $\ket{v_1\cdots v_k}\equiv\ket{v}$.  Similarly we define exponentiation such that $\ket{v_j^2}\equiv\ket{v_jv_j}$.  Then applying the $k-1$ controlled not gates to the state and then the Hadamard on the transform first qubit we find that we can write
\begin{align}
{\rm GHZ}^{-1}&\left(e^{-i\theta_1 X} \otimes \cdots \otimes e^{-i\theta_k X}\ket{0}^{k} \ket{\psi} - \prod_{j=1}^k \cos(\theta_j)\ket{0^k}\ket{\psi} -(-i)^k \prod_{j=1}^k \sin(\theta_j) \ket{1^{k}}\ket{\psi}\right) \nonumber\\
&= (H\otimes \openone)\sum_{0^k\ne v \ne 1^k} \prod_{j=1}^k \cos(\theta_j)^{1-v_j}\sin(\theta_j)^{v_j} (-i)^{\sum_{q=1}^k v_q}\ket{v \oplus 0(v_1)^{k-1}}\ket{\psi}\nonumber\\
&=\frac{\sum_{0^k\ne v \ne 1^k} \prod_{j=1}^k \cos(\theta_j)^{1-v_j}\sin(\theta_j)^{v_j} (-i)^{\sum_{q=1}^k v_q}(\ket{v \oplus (v_1)^{k}}+(-1)^{v_1}\ket{v\oplus (v_1+1)v_1^{k-1}})}{\sqrt{2}}\ket{\psi}\nonumber\\
&=\frac{\sum_{0^k\ne v \ne 1^k} \prod_{j=1}^k \cos(\theta_j)^{1-v_j}\sin(\theta_j)^{v_j} (-i)^{\sum_{q=1}^k v_q}(\ket{0}\ket{v_2\cdots v_k \oplus (v_1)^{k-1}}+(-1)^{v_1}\ket{1}\ket{v_2\cdots v_k\oplus v_1^{k-1}})}{\sqrt{2}}\ket{\psi}\nonumber\\
&:=\frac{\sqrt{1-\beta^2}}{\sqrt{2}} \left(\ket{0}\ket{\chi_0} + \ket{1}\ket{\chi_1} \right)\ket{\psi}\label{eq:GHZtrans},
\end{align}
for some $\beta$ that is a function of $\theta_1,\ldots,\theta_k$.  The quantity $\beta$ is included in order to ensure that the resulting state has less than unit length.
Since $v_2\cdots v_k \ne (v_1)^{k-1}$ for any of the states in~\eq{GHZtrans}, ${\rm GHZ}^{-1} \ket{\phi}$ is orthogonal to ${\rm span}(\ket{0^k}, \ket{1}\ket{0^{k-1}})$, which means that we can use~\eq{GHZtrans} to write for $V_0= \exp(-i\tan^{-1}(\prod_j \tan(\theta_j))X)$
\begin{align}
\PAR(\theta_1,\ldots, \theta_k):\ket{0}^k\ket{\psi}\mapsto \frac{\ket{0}\left(\beta\ket{0}^{k-1} V_0 \ket{\psi} + \sqrt{1-|\beta|^2}\ket{\chi_0}\ket{\psi} \right)+\ket{1}\left(\beta\ket{0}^{k-1} V_0^\dagger \ket{\psi} + \sqrt{1-\beta^2}\ket{\chi_1}\ket{\psi} \right)}{\sqrt{2}}.\label{eq:PARfinal}
\end{align}
Since the overall state must be normalized to unit length, we can infer from~\eq{betaeq} that \begin{equation}
|\beta|=\sqrt{\prod_{j=1}^k \cos^2(\theta_j) + \prod_{j=1}^k \sin^2(\theta_j)}.
\end{equation}

Equation \eq{PARfinal} is of the form required by~\cor{RUS}, which implies that we can use oblivious amplitude amplification to deterministically prepare the state
$$
\beta\ket{0}^{k-1} V_0 \ket{\psi} + \sqrt{1-\beta^2}\ket{\chi_0}\ket{\psi},
$$
using $3$ PAR rotations and a constant sized Clifford circuit.  If we then measure the first $k-1$ qubits of this state to be $0^{k-1}$ then $V_0\ket{\psi}$ will be performed as required; whereas if we measure any other outcome the identity gate will be applied to $\ket{\psi}$.  The probability of implementing $V_0$ is therefore $\beta^2 = \prod_{j=1}^k \cos^2(\theta_j) + \prod_{j=1}^k \sin^2(\theta_j)$.  The circuit is therefore a Repeat-Until-Success circuit with the success probability claimed by~\lem{PAR}.
\end{proofof}

\section{Bounds on Mean and Variance of Cost of RUS Circuits}\label{app:meanvar}
Of course, if we nest several PAR and Gearbox circuits then we have to be able to bound the number of repetitions needed for success to occur with high probability.  This can be done using Chebyshev's inequality given the expectation value and the variance of the number of repetitions of the RUS circuits used.  These properties are summarized in the following lemma.

\begin{lemma}\label{lem:RUS}
Let $V$ be an RUS circuit composed of $k$ RUS circuits whose success probabilities are independent.  Let $x_1,\ldots,x_k$ be random variables that describe the number of attempts of each of the $k$ RUS circuits that are needed in order to achieve a successful measurement of each of the $k$ RUS circuits and let $N$ be a random variable describing total number of repetitions of $V$ needed for success.  The total number of repetitions of $x_1,\ldots,x_k$, $R_x$, that are needed obeys
\begin{align}
\mathbb{E} (R_x) &= \mathbb{E}(N) \sum_{j=1}^k \mathbb{E}(x_j)\\
\mathbb{V} (R_x) &= \mathbb{E}(N)\sum_{j=1}^k \mathbb{V}(x_j)+\mathbb{V}(N)(\sum_{j=1}^k \mathbb{E}(x_j))^2
\end{align}
\end{lemma}
\begin{proof}
Expanding $R_x$,
\begin{equation}
R_x := \sum_{i=1}^N \sum_{j=1}^k x_j = \sum_{i=1}^\infty \sum_{j=1}^k x_j \openone_{i\le N}= \sum_{i=1}^\infty \sum_{j=1}^k (\chi_j+\mathbb{E}(x_j))\openone_{i\le N},
\end{equation}
where each $\chi_j$ is a zero--mean random variable such that $x_j = \chi_j +\mathbb{E}(x_j)$ and $\openone_{i\le N}$ is the indicator function that is $1$ for $i\le N$ and $0$ otherwise.

The random variables $N$ and $\chi_j$ are independent and  each $\chi_j$ has zero--mean, which means that
\begin{equation}
\mathbb{E}(R_x) = \mathbb{E}( \sum_{i=1}^\infty \sum_{j=1}^k (\chi_j+\mathbb{E}(x_j))\openone_{i\le N}) = \mathbb{E}(\sum_{i=1}^\infty \openone_{i\le N})\sum_{j=1}^k \mathbb{E}(x_j)= \mathbb{E}(N)\sum_{j=1}^k \mathbb{E}(x_j).
\end{equation}
The variance is a little more difficult to compute
\begin{equation}
\mathbb{V}(R_x) = \mathbb{E}(R_x^2) - \mathbb{E}(R_x)^2.
\end{equation}
Expanding this out we find
\begin{align}
\mathbb{V}(R_x) = \mathbb{E}\left(\left[\sum_{i=1}^\infty \sum_{j=1}^k (\chi_j +\mathbb{E} x_j)\openone_{i\le N}\right]\left[\sum_{i'=1}^\infty \sum_{j'=1}^k (\chi_{j'} +\mathbb{E} x_{j'})\openone_{{i'}\le N'}\right] \right) - \mathbb{E}(N)^2(\sum_{j=1}^k \mathbb{E}(x_j))^2.\label{eq:variance}
\end{align}
Since $\mathbb{E}(\chi_j)=0$, many of the terms in~\eq{variance} are zero.  In fact, any term in that is not of the form $\chi_{j}^2$ or $\mathbb{E}(\chi_j)\mathbb{E}(\chi_k)$ is zero.  After this observation we see using the independence of each random variable that
\begin{align}
\mathbb{V}(R_x)&=\mathbb{E}\left(\left[\sum_{i=1}^\infty \sum_{j=1}^k \chi_j^2 \openone_{i\le N}\right]+\left[\sum_{i=1}^\infty \sum_{j=1}^k \mathbb{E} x_j\openone_{i\le N}\right]\left[\sum_{i'=1}^\infty \sum_{j'=1}^k \mathbb{E} x_{j'}\openone_{{i'}\le N'}\right] \right) - \mathbb{E}(N)^2(\sum_{j=1}^k \mathbb{E}(x_j))^2\nonumber\\
&= \mathbb{V}(x_j)\mathbb{E}(N) + \mathbb{E}(N^2)(\sum_{j=1}^k\mathbb{E}(x_j))^2-\mathbb{E}(N)^2 (\sum_{j=1}^k\mathbb{E}(x_j))^2\nonumber\\
&= \mathbb{V}(x_j)\mathbb{E}(N) + \mathbb{V}(N) (\sum_{j=1}^k\mathbb{E}(x_j))^2
\end{align}
\end{proof}

This provides estimates of the mean and the variance of the number of times that the RUS circuits have to be applied in order to achieve the desired rotation.  It is easy to go from these results to find estimates for the number of $T$ gates needed to implement either PAR or \GB.  We focus on the number of $T$ gates required because of their importance in fault tolerant quantum computing wherein $T$ gates are often considered to be the most costly operations because they can require several rounds of magic state distillation.
\begin{corollary}\label{cor:cost}
Let $T_1,\ldots, T_k$ be the probability distributions that describe the number of $T$ gates needed to implement the unitaries $\phi_1,\ldots,\phi_k$, $T_{\PAR(\phi_1,\ldots,\phi_k)}$ be the cost of the PAR circuit (acting on $\ket{0}$) and $T_{\GB(\phi_1,\ldots,\phi_k)}$ be the number of $T$ gates needed to implement the gearbox circuit then
\begin{align}
\mathbb{E}(T_{\PAR(\phi_1,\ldots,\phi_k)})&=\frac{4(k-1)+\sum_{j=1}^k \mathbb{E}(T_j)}{P},\\
\mathbb{V}(T_{\PAR(\phi_1,\ldots,\phi_k)})&= \frac{(4(k-1)+\sum_{j=1}^k \mathbb{E}(T_j))^2(1-P)}{P^2}+\frac{\sum_{j=1}^k \mathbb{V}(T_j)}{P},
\end{align}
where $P=\prod_{j=1}^k \cos(\phi_j)^2 +\prod_{j=1}^k \sin(\phi_j)^2$.  Similarly,
\begin{align}
\mathbb{E}(T_{\GB(\phi_1,\ldots,\phi_k)})&=\frac{4(k-1)+2\sum_{j=1}^k \mathbb{E}(T_j)}{P},\\
\mathbb{V}(T_{\GB(\phi_1,\ldots,\phi_k)})&= \frac{(4(k-1)+2\sum_{j=1}^k \mathbb{E}(T_j))^2(1-P)}{P^2}+\frac{2\sum_{j=1}^k \mathbb{V}(T_j)}{P},
\end{align}
where $Q=\prod_{j=1}^k \cos(\phi_j)^4 +\prod_{j=1}^k \sin(\phi_j)^4$.
\end{corollary}
\begin{proof}
Let us first begin with PAR.  There are two sources of $T$ gates in PAR: $T$ gates from the $k$--controlled $X$ gate and $T$ gates from implementing the $k$ rotations:
\begin{equation}
\mathbb{E}(T_{\PAR(\phi_1,\ldots,\phi_k)})= 4(k-1) \mathbb{E}(N) +\mathbb{E}(\sum_{i=1}^N \sum_{j=1}^k T_j)= (4(k-1)+\sum_{j=1}^k \mathbb{E}(T_j))\mathbb{E}(N),
\end{equation}
here I have assumed that the Toffoli gate construction of~\cite{JWM+12} is used to implement the $k$--controlled $X$ gate. $N$ obeys a geometric distribution with $p=P$, which means that $\mathbb{E}(N)=1/P$ and hence we find by substituting $x_j\rightarrow T_j$ in~\lem{RUS} that
\begin{equation}
\mathbb{E}(T_{\PAR(\phi_1,\ldots,\phi_k)})=  \frac{4(k-1)+\sum_{j=1}^k \mathbb{E}(T_j)}{P}.
\end{equation}
Similarly since the variance of the geometric distribution is $(1-P)/P^2$,
\begin{align}
\mathbb{V}(T_{\PAR(\phi_1,\ldots,\phi_k)})&=   \mathbb{V}(\sum_{i=1}^N \sum_{j=1}^k T_j+4(k-1))\nonumber\\
&= \frac{(4(k-1)+\sum_{j=1}^k \mathbb{E}(T_j))^2(1-P)}{P^2}+\frac{\sum_{j=1}^k \mathbb{V}(T_j)}{P}.
\end{align}

The analogous formulas for $\GB$ can be found by noting that the main difference between the $T$ count of $\GB$ and $\PAR$ is that each $\phi_j$ is repeated twice in \GB, which leads to the costs claimed above.
\end{proof}

\section{Non--RUS gearbox circuits}\label{app:GB}
$\GB$~can be converted into a non--deterministic circuit that requires no online rotations as given below.  Such circuits are particularly important in cases where qubits are inexpensive and the architecture allows any parallel operations to be performed simultaneously.  This circuit may be preferable to $\PAR(\phi_1,\phi_1)$ for approximately squaring input rotations because the circuit requires only online Clifford operations, which are typically inexpensive in fault tolerant implementations.  Such circuits also provide a lower $T$--depth method for implementing $\GB^{\circ k}(x)$ than that used in~\sec{square} or~\cite{WK13} because the gearbox circuit is insensitive to the sign of the input rotation angle and also because most of the required operations act upon $\ket{0}$.  The latter case is significant because under--rotations can be corrected by applying a $Z$ gate if required.
\begin{lemma}
There exists a non--deterministic version of the gearbox circuit $\GB(\phi_1)$ wherein all of the rotations can be performed offline that has success probability $\frac{1}{2}(\cos^4(\phi_1)+ \sin^4(\phi_1))$ and upon failure either performs $e^{i \arctan(\tan^2(\phi_1)X)}$ or applies a Clifford operation.
\end{lemma}
\begin{proof}
The circuit is

\begin{align}
 &\Qcircuit @R 1em @C 1.5em {
\lstick{e^{-i X\phi_1}\ket{0}}	&\qw		&\ctrl{1}	&\gate{H}	&\meter\\
\lstick{e^{-i X\phi_1}\ket{0}}	&\ctrl{1}	&\targ		&\qw		&\meter\\
\lstick{\ket{\psi}}			&\gate{-iX}	&\qw		&\qw		&\qw\\
}\label{eq:composedcirc2}
\end{align}
We have from~\thm{PAR} that the operations affecting the topmost qubit simply act to teleport the rotation to the middle qubit upon success.  In this case, we want the rotation direction to flip which happens with probability $1/2$ for $\PAR(\phi_1)$.  If this operation is successful then the operations on the bottom two qubits are equivalent to those in $\GB(\phi_1)$, which succeeds with probability $(\cos^4(\phi_1)+ \sin^4(\phi_1))$~\cite{WK13}.  The claim regarding the overall success probability then follows from noting that the measurement success probabilities are independent and hence the total success probability is the product of the two probabilities.

If the first measurement yield $``1"$ then the circuit will deviate from the original gearbox circuit. This results in the following effective transformation for the bottom two qubits.
\begin{align}
e^{-i\phi_1 X}\ket{\psi} \mapsto \cos^2(\phi_1)\ket{0}\ket{\psi}-\sin^2(\phi_1) \ket{0}(-iX)\ket{\psi} -i \cos(\phi_1)\sin(\phi_1) (\ket{1}\ket{\psi}+ \ket{1}(-iX)\ket{\psi}).
\end{align}
If the middle qubit is measured to be $0$ then $e^{i\arctan(\tan^2(\phi_1))X}$ is applied on the bottom most qubit, which is a rotation in the opposite direction from what is intended.  Alternatively if the middle qubit is measured to be $1$ then a Clifford operation, $e^{-i\pi X/4}$, is applied to $\ket{\psi}$.
\end{proof}
The possibility of failure can also be removed from these circuits if the user has access to $\phi_1$ as well as $\{\phi_j:j=2,\ldots \infty\}$ wherein
\begin{equation}
\phi_{j+1} = \tan^{-1}\left(\sqrt{\tan(2\tan^{-1}(\tan^2 (\phi_j)))} \right).
\end{equation}
This is because $\GB(\phi_{j+1}) = 2\GB(\phi_{j})$ and thus, similar to~\cite{JWM+12}, if the direction of the rotation is opposite to that initially intended then $\GB(\phi_j)$, for any $j\ge 2$, will correct the sum of all failures due to prior attempts:
\begin{equation}
\GB(\phi_j) - \sum_{j=1}^{j-1} \GB(\phi_2) = \phi_1.
\end{equation}
Obviously, the cost of preparing such states using reversible logic or RUS synthesis may be quite high so we do not advocate this approach in general.  However, it is interesting to note that many of the costs of RUS synthesis can be reduced if we were to posit the existence of a more complex resource state factory.

%Thus the generalized PAR circuits satisfy many of the same properties as multiplication and in fact reduces to multiplication in the limit where the input rotation angles tend to zero.  This last property, non--linearity, is certainly a drawback if one is trying to emulate multiplication; however, it can also be used as a feature in other applications.

%The gearbox circuits are more analogous to a square multiplier which maps $(x,y)\mapsto x^2y^2$ as seen below

%Again, these properties are very similar to those of the map $(x,y)\mapsto x^2y^2$ with the notable exception of the non--linearity, periodicity and related properties.  An interesting feature of these functions is the orthogonality property.  This feature means that such approximate square waves can also be used to generate Haar--like wavelets, which would allow an entirely new class of methods to implement a function of input rotation angles than that considered below.

\bibliographystyle{unsrt}
%\bibliography{QNA}

\end{document}